\documentclass[10pt,journal,compsoc]{IEEEtran}

\ifCLASSOPTIONcompsoc
  \usepackage[nocompress]{cite}
\else
  \usepackage{cite}
\fi

\ifCLASSINFOpdf
   \usepackage[pdftex]{graphicx}
   \graphicspath{{./figs/}}
\else
   \usepackage[dvips]{graphicx}
   \graphicspath{{./figs/}}
\fi
\usepackage{microtype}                
\PassOptionsToPackage{warn}{textcomp}  
\usepackage{textcomp}                  
\usepackage{mathptmx}                  
\usepackage{times}                     
         
\usepackage{cite}                      

\usepackage{tabu}                      
\usepackage{booktabs}          
\usepackage{multirow}

\newcommand{\para}[1]  {\vspace{0mm}\noindent{\textbf{#1}}}

\newcommand{\etal}{{et al.}}
\newcommand{\eg}{{e.g.}}
\newcommand{\wrt}{{w.r.t.}}

\newcommand{\ie}{{i.e.}}

\usepackage{amssymb}
\usepackage{amsmath}

\usepackage{amsthm}

\usepackage{amsfonts}
\usepackage[ruled]{algorithm2e}
\usepackage{tikz}
\usetikzlibrary{arrows}

\usepackage{color}

\newcommand{\update}[1]{{\textcolor{black}{#1}}}

\newcommand {\mm}[1]        {\ifmmode{#1}\else{\mbox{\(#1\)}}\fi}

\newcommand{\onevec}[1]     {\mm{\mathbf{1}_{#1}}}

\newcommand{\Rspace}        {\mm{\mathbb{R}}}

\newcommand{\Mspace}        {\mm{\mathbb{M}}}

\newcommand{\lca}        {\mm{\mathrm{lca}}}

\newcommand{\parent}        {\mm{\mathrm{parent}}}

\newcommand{\Ccal}        {\mm{\mathcal{C}}}

\newcommand{\HC}        {\mm{\mathsf{Heated\,Cylinder}}}
\newcommand{\IF}        {\mm{\mathsf{Ionization\,Front}}}
\newcommand{\VS}        {\mm{\mathsf{Vortex\,Street}}}
\newcommand{\IS}        {\mm{\mathsf{Isabel}}}
\newcommand{\CL}        {\mm{\mathsf{Cloud}}}
\newcommand{\VF}        {\mm{\mathsf{Viscous\,Fingering}}}
\newcommand{\UCF}        {\mm{\mathsf{Unsteady\,Cylinder\,Flow}}}

\newcommand{\R} {\mathbb{R}}

\usepackage{mathptmx}
\DeclareMathAlphabet{\mathcal}{OMS}{cmsy}{m}{n}
\usepackage{mathtools}
\DeclareMathDelimiter{(}{\mathopen} {operators}{"28}{largesymbols}{"00}
\DeclareMathDelimiter{)}{\mathclose}{operators}{"29}{largesymbols}{"01}

\usepackage[hidelinks]{hyperref}

\usepackage{letltxmacro}
\LetLtxMacro{\originaleqref}{\eqref}
\renewcommand{\eqref}{Eq.~\originaleqref}

\usepackage{amsthm}

\newtheorem{theorem}{Theorem}
\newtheorem{corollary}{Corollary}
\newtheorem*{corollary*}{Corollary}
\newtheorem*{theorem*}{Theorem}
\newtheorem{lemma}{Lemma}
\newtheorem{proposition}{Proposition}
\newtheorem{remark}{Remark}

\newcommand{\argmax}{\mathrm{argmax}}

\begin{document}

\title{Flexible and Probabilistic Topology Tracking with Partial Optimal Transport}

\author{
Mingzhe Li, 
Xinyuan Yan, 
Lin Yan, 
Tom Needham,
Bei Wang
\IEEEcompsocitemizethanks{\IEEEcompsocthanksitem 
M. Li, X. Yan, and B. Wang are with the University of Utah, Salt Lake City, UT, 84112. \protect\\
E-mails: mingzhe.li@utah.edu, \{xinyuan.yan, beiwang\}@sci.utah.edu
\IEEEcompsocthanksitem 
L. Yan is with the Iowa State University,  Ames, IA, 50011.\protect\\
E-mail: linyan@iastate.edu
\IEEEcompsocthanksitem 
T. Needham is with the Florida State University, Tallahassee, FL, 32306.\protect\\
E-mail: tneedham@fsu.edu}
\thanks{Manuscript received April 19, 2005; revised August 26, 2015.}}

\markboth{Journal of \LaTeX\ Class Files,~Vol.~14, No.~8, August~2015}%
{Li \MakeLowercase{\textit{et al.}}: Flexible Topological Tracking}

\IEEEtitleabstractindextext{
\begin{abstract}
In this paper, we present a flexible and probabilistic framework for tracking topological features in time-varying scalar fields using merge trees and partial optimal transport.
Merge trees are topological descriptors that record the evolution of connected components in the sublevel sets of scalar fields.
We present a new technique for modeling and comparing merge trees using tools from partial optimal transport. 
In particular, we model a merge tree as a measure network, that is, a network equipped with a probability distribution, and define a notion of distance on the space of merge trees inspired by partial optimal transport. Such a distance offers a new and flexible perspective for encoding intrinsic and extrinsic information in the comparative measures of merge trees. 
More importantly, it gives rise to a partial matching between topological features in time-varying data, thus enabling flexible topology tracking for scientific simulations. 
Furthermore, such partial matching may be interpreted as probabilistic coupling between features at adjacent time steps, which gives rise to probabilistic tracking graphs. 
We derive a stability result for our distance and provide numerous experiments indicating the efficacy of our framework in extracting meaningful feature tracks.

\end{abstract}
}

\maketitle

\IEEEdisplaynontitleabstractindextext

\IEEEpeerreviewmaketitle

\IEEEraisesectionheading{\section{Introduction}
\label{sec:introduction}}

\IEEEPARstart{F}{eature} extraction and tracking for time-varying data play an important role in scientific visualization. 
Over the past two decades, topology-based techniques have been successfully applied to study the evolution of features of interest, which is at the core of many scientific applications, including   combustion~\cite{BremerWeberTierny2011}, climatology~\cite{EngelkeMasoodBeran2021}, and astronomy~\cite{FriesenAlmgrenLukic2016}. 
In particular, topology-based techniques utilize topological descriptors such as persistence diagrams and merge trees for feature extraction and tracking in scalar field data;   see~\cite{HeineLeitteHlawitschka2016,YanMasoodSridharamurthy2021} for surveys.  

In this paper, we present a novel and flexible framework for tracking features (\ie, critical points) in time-varying scalar fields by combining merge trees with partial optimal transport. 
Merge trees are topological descriptors that record the evolution of connected components in the sublevel sets of scalar fields. 

The theory of optimal transport studies distances between probability distributions. 
In its simplest form, it studies the transportation problem of moving a pile of dirt (i.e., a probability distribution) to a target pile (i.e., another probability distribution) with minimum cost. 
The classic Wasserstein distance in optimal transport  is thus called the Earth mover's distance. 
Whereas classic optimal transport preserves the total amount of dirt (i.e., the total mass) to be transported, partial optimal transport requires a fraction of the total mass to be transported. The contributions of this paper include: 
\begin{itemize}
\item We present a new technique for modeling and comparing merge trees using tools from partial optimal transport. In particular, we model a merge tree as a measure network (that is, a network equipped with a probability distribution) and define a partial fused Gromov-Wasserstein distance between a pair of merge trees. 
\item We show that such a distance offers a new and flexible way to encode intrinsic and extrinsic information in the comparative measures of merge trees. We also derive a stability result for our distance under a restrictive setting.  
\item Most importantly, we demonstrate via extensive experiments that such a distance gives rise to a partial matching between topological features in time-varying data, thus enabling flexible topology tracking for scientific simulations. 
\item Finally, the partial optimal transport provides a probabilistic coupling between features at adjacent time steps, which are then visualized by weighted tracks from probabilistic tracking graphs. 
\end{itemize}
Furthermore, our implementation is open source\footnote{ 
\url{https://github.com/tdavislab/GWMT}}, and comes with a video that demonstrates the probabilistic tracking graph. 
 	
\para{Overview.} 
After reviewing related work on optimal transport and topology-based feature tracking in~\autoref{sec:related-work}, we review the technical background of merge trees, measure networks, and various distances used in (partial) optimal transport in~\autoref{sec:background}. 
We then describe our novel feature-tracking framework in \autoref{sec:method}. In particular, we introduce a new distance---\emph{partial fused Gromov-Wasserstein distance}---in~\autoref{sec:pfgw} and describe its theoretical properties (\autoref{sec:stability-overview}).  
We demonstrate the utility of our framework with extensive experiments and comparisons with the state-of-the-art (\autoref{sec:results}).   
A direct consequence of our framework is that it enables richer representations of tracking graphs, referred to as \emph{probabilistic tracking graphs}, for which we give a visual demonstration in~\autoref{sec:tracking-graphs}.

\section{Related Work}
\label{sec:related-work}

\para{Optimal transport and Gromov-Wasserstein distance.} 
This paper builds upon the \emph{Gromov-Wasserstein (GW) distance}, a tool from optimal transport for deriving probabilistic correspondences between nodes of different networks. 
Specifically, we use GW distance to study \emph{merge trees}, which are topological descriptors of scalar fields; see~\autoref{sec:background} for formal definitions. 
The GW distance was introduced by M\'{e}moli~\cite{Memoli2007,Memoli2011} as a way to compare metric measure spaces (\ie, compact metric spaces endowed with probability measures), with a view to shape analysis applications. More recently, this framework was extended to allow comparisons between networks endowed with kernel functions that are not necessarily metrics \cite{PeyreCuturiSolomon2016,ChowdhuryMemoli2019}. 
The GW distance has become an important tool in machine learning applications, such as graph matching and partitioning \cite{xu2019gromov,xu2019scalable,chowdhury2021generalized}, natural language processing \cite{Alvarez-MelisJaakkola2018}, and alignment of  multiomics data \cite{demetci2020gromov}. 

A number of recent works have focused specifically on applications of GW distance to merge trees. 
Combining a Riemannian interpretation of GW distance developed in \cite{Sturm2012,ChowdhuryNeedham2020} with matrix sketching techniques, Li \etal \cite{LiPalandeYan2023} introduced a pipeline for finding structural representatives among a set of merge trees.
 In \cite{curry2022decorated}, GW techniques were combined with theory developed in \cite{GasparovicMunchOudot2019} in order to give an estimate of an \emph{interleaving distance} on the space of merge trees. 
 Theoretical properties of a refined generalization of GW distance between merge tree-like objects called \emph{ultra dissimilarity spaces} were studied in \cite{MemoliMunkWan2023}. 
 
In this paper, we present a novel distance between merge trees, called the \emph{partial fused Gromov-Wasserstein  (pFGW) distance}, which is built upon variants of the GW pipeline, including the Fused Gromov-Wasserstein distance~\cite{VayerChapelFlamary2020} and partial optimal transport~\cite{ChapelAlayaGasso2020}.  

\para{Merge tree comparisons.} 
A number of recent works have studied the distances between merge trees or, more generally, Reeb graphs. 
For instance, the functional distortion distance~\cite{bauer2014measuring}, the interleaving distance~\cite{morozov2013interleaving,de2016categorified}, the Gromov-Hausdorff distance~\cite{AgarwalFoxNath2018}, the Reeb graph edit distance~\cite{Bauer2018b, DiFabio2016}, the merge tree matching distance~\cite{Bollen2023}, and the distance based on branch decomposition~\cite{BeketayevYeliussizovMorozov2014} are equipped with some desirable theoretical properties, including stability; see~\cite{YanMasoodSridharamurthy2021, bollen2022reeb} for surveys. 
Our work provides a new stability result of the GW distance between merge trees with theoretical justifications.
 
\para{Topology-based feature tracking.} 
Topological techniques have been used for feature extraction and tracking in scalar fields~\cite{YanMasoodSridharamurthy2021} and vector fields~\cite{PostVrolijkHauser2003, BujackYanHotz2020}. 
 
Topology has been used to track features for time-varying scalar fields by solving an explicit correspondence problem. 
A number of topological descriptors have been used for feature tracking, including persistence diagrams, merge trees, contour trees, Reeb graphs, extremum graphs, and Morse complexes; see~\cite[Sec.~7.1]{YanMasoodSridharamurthy2021} for a survey.
Recently, persistence diagrams and an extension of the Wasserstein metric have been used to perform topology tracking~\cite{SolerPlainchaultConche2018, SolerPetitfrereDarche2019}. 
A metric on the space of merge trees was recently introduced~\cite{PontVidalDelon2021} based on the $L_2$-Wasserstein distance between extremum persistence diagrams.  
Yan {\etal}~\cite{YanMasoodRasheed2022} performed geometry-aware comparisons of merge trees using labeled interleaving distances.
Their framework uses a labeling step to find a correspondence between the critical points of two merge trees, and integrates geometric information of the data domain in the labeling process~\cite{YanMasoodRasheed2022}. 
Instead, our distance computation utilizes information from the data domain within the distances themselves. 

Our pFGW distance applies to any task involving merge tree comparisons; in this paper, we focus on feature tracking in time-varying scalar fields using merge trees. 
A popular approach to obtain the correspondence between features is to compute the overlap between regions or volumes surrounding the features. For instance, Lukasczyk \etal~captured the evolution of superlevel set components~\cite{LukasczykWeberMaciejewski2017, lukasczyk2019dynamic} based on the overlaps between  their corresponding regions.
Saikia \etal~\cite{saikia2017global, saikia2017fast} presented a strategy for topological feature tracking with merge trees called Global Feature Tracking (GFT). 
Their strategy determines the similarity of subregions segmented by merge trees at adjacent time steps, based on the overlap size between two regions, and the similarity between histograms of scalar values within each region. 
In GFT, the information of a critical point includes its subtree, whereas our work considers the relation between every pair of critical points in the merge tree. 
Furthermore, GFT uses the segmentation of scalar fields to compare the overlapping subtree regions, which can be memory-consuming. 

Recent works~\cite{SolerPlainchaultConche2018, SolerPetitfrereDarche2019,PontVidalDelon2021} have utilized persistence diagrams for feature tracking. 
Alternatively, these approaches could be considered as solving an assignment problem using branch decompositions of merge trees. 
Such assignment problems are closely related to (partial) optimal transport~\cite{divol2021understanding}. 

In particular, Soler \etal~introduced the Lifted Wasserstein Matcher (LWM)~\cite{SolerPlainchaultConche2018} framework, where features are tracked based on the optimal matching between persistence diagrams under the Wasserstein distance. 
The cost of matching a pair of points in the persistence diagram is a weighted linear combination of: (a) the geometric distances between the extrema involved in the persistence pairs, and (b) the differences between the birth and death coordinates of the points in the diagram. 
The cost of matching a point to its diagonal projection (causing disappearances and appearances of features) is a weighted linear combination of: (a)  the geometric distance between the critical points associated with the point in the diagram, and (b) the birth and death coordinates of the point.  
Both LWM and pFGW approaches make use of geometric locations of critical points. LWM encodes topological information via birth and death coordinates of the points in the persistence diagram, whereas pFGW encodes topological constraints on the matched critical points via their relations within merge trees.
Whereas LWM solves an assignment problem deterministically, pFGW gives rise to probabilistic matching between features.
Another interesting feature of our approach is that we are able to derive a stability result (Theorem \ref{thm:stability-main}), which has so far not been established for some of the other methods (\eg, \cite{PontVidalDelon2021}) in the literature.

Although this paper focuses on feature tracking in scalar fields, we review feature tracking in vector fields briefly, which  also aims to associate features from one time step to the next, and to detect topological events.  
Helman and Hesselink~\cite{HelmanHesselink1989, helman1990surface} tracked critical points in vector fields over time, and Wischgoll \etal~\cite{wischgoll2001tracking} tracked closed streamlines and detected bifurcations.
Tricoche \etal~\cite{tricoche2001topology, tricoche2002topology} provided critical point tracking using spacetime grids. 
Theisel and Seidel~\cite{theisel2003feature} introduced  Feature Flow Fields (FFF), followed by stable~\cite{weinkauf2011stable} and combinatorial~\cite{reininghaus2012efficient} variants.
See~\cite[Sec.~4.1]{BujackYanHotz2020} for a survey. 

\para{Feature tracking graphs} have been used to visualize the evolution (i.e.,~births, deaths, merging and splitting) of topological features over time (e.g.,~\cite{LukasczykAldrichSteptoe2017, widanagamaachchi2012interactive}). 
A probabilistic tracking graph may arise when the edges in the graph are equipped with weights that correspond to the amount of spatial overlap between the connected features. 
In our setting, we introduce a different notion of a probabilistic feature tracking graph, which uses the coupling probabilities between features across time steps. 
These probabilities are derived based on the locations of critical points as well as their structural relations captured by the merge trees.

\section{Technical Background}
\label{sec:background}

We combine ingredients from diverse areas: topology in visualization, optimal transport, and measure theory. 
We first review the merge tree of a scalar field in topology-based visualization (\autoref{sec:mt}). 
We then introduce concepts from optimal transport and measure theory, including measure networks  (\autoref{sec:measure-networks}), Wasserstein distance, Gromov-Wasserstein (GW) distance, and fused Gromov-Wasserstein (FGW) distance (\autoref{sec:fgw}). 
We further discuss the partial Wasserstein and  partial GW distances within partial optimal transport, which set up the foundation for our new partial FGW distance (\autoref{sec:pfgw}). 

\subsection{Merge Trees}
\label{sec:mt}

Let $f: \Mspace \to \Rspace$ be a scalar field defined on the domain of interest $\Mspace$, where $\Mspace$ can be a manifold or a subset of $\Rspace^d$.
For our experiments, $\Mspace \subset \Rspace^2$ or $\Rspace^3$. 
Merge trees capture the connectivity among the \emph{sublevel sets} of $f$, \ie, $\Mspace_a = f^{-1}(-\infty, a]$. 
Formally, two points $x, y \in \Mspace$ are considered to be \emph{equivalent}, denoted by $x \sim y$, if $f(x) = f(y)=a$, and $x$ and $y$ belong to the same connected component of a sublevel set $\Mspace_a$. 
The \emph{merge tree}, $T(\Mspace, f) = \Mspace/{\sim}$, is the quotient space obtained by gluing together points in $\Mspace$ that are equivalent under the relation $\sim$; see~\autoref{fig:mt} for an example. 

\begin{figure}[!ht]
    \centering
    \includegraphics[width=0.98\columnwidth]{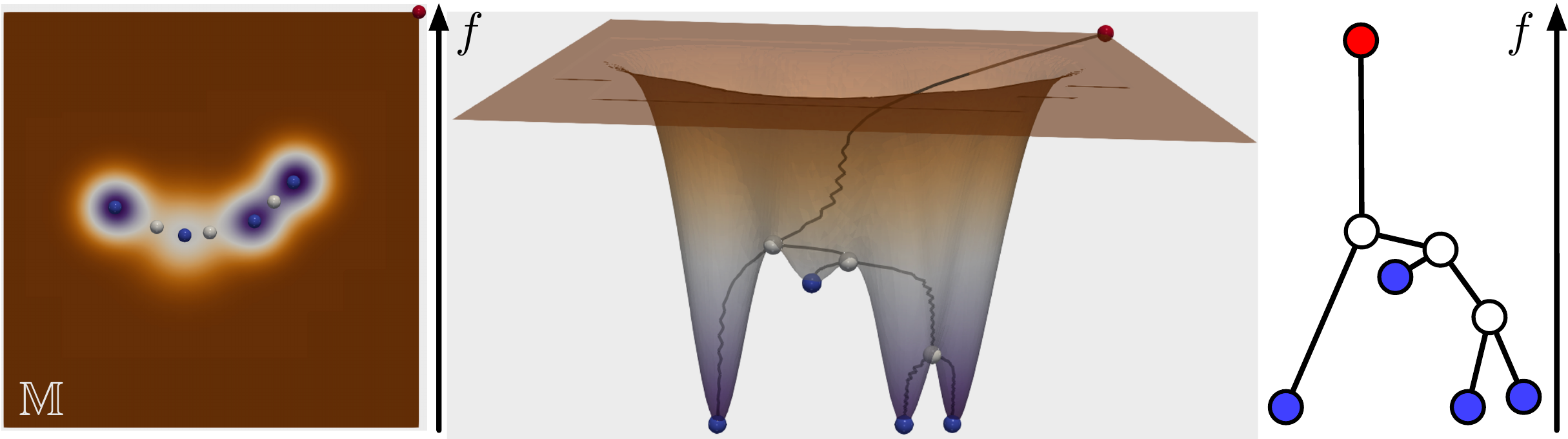}
    \vspace{-2mm}
    \caption{An example of a merge tree from a height field $f: \Mspace \to \Rspace$ defined on a 2D domain. From left to right: (a) 2D scalar field visualization with local minima in blue, saddles in white, and local maxima in red; (b) a merge tree embedded in the graph of the scalar field; and (c) an abstract (straight-line) visualization of a merge tree as a rooted tree equipped with the height function.} 
    \label{fig:mt}
\end{figure}

The construction of a merge tree for a given $f:\Mspace \to \Rspace$ is described procedurally as follows: we sweep the function value $a$ from $-\infty$ to $\infty$, and we create a new branch originating at a leaf node for each local minimum of $f$. 
As $a$ increases, such a branch is extended as its corresponding component in $\Mspace_a$ grows until it merges with another branch at a saddle point. 
Assuming $\Mspace$ is connected and $f$ achieves a unique global maximum, then all branches eventually merge into a single component, which corresponds to the root of the tree. 
For a given merge tree, leaves, internal nodes, and root node represent the minima, merging saddles, and global maximum of $f$, respectively. 
\autoref{fig:mt} displays a height function $f:\Mspace \subset \Rspace^2 \to \Rspace$ in (a),  together with its corresponding merge tree embedded in the graph of the scalar field, \ie,~$\{(x,f(x)):x\in \Mspace\}$ in (b).  
Abstractly, a merge tree $T$ is a rooted tree equipped with $f$ restricted to its node set, $f: V \to \Rspace$, as shown in (c). 

\subsection{Measure Networks}
\label{sec:measure-networks}

A finite graph $G$ may be represented as a \emph{measure network}~\cite{ChowdhuryMemoli2019} using a triple $(V, p, W)$: $V$ is the set of $n$ nodes in the graph, $p$ is a probability measure supported on the nodes of $G$, and $W \in \R^{|V| \times |V|}$ is a matrix that encodes  relational information between the nodes. 
For example, $W$ may be a weighted adjacency matrix~\cite{xu2019scalable}, a graph Laplacian~\cite{ChowdhuryNeedham2020}, or a matrix of graph distances~\cite{hendrikson2016using}. 
Without prior knowledge about $G$, $p$ is typically taken to be uniform; that is, $p(x) = 1/n$, for each $x \in V$. 
We represent $p$ as a vector of size $n$, $p = \frac{1}{n} \onevec{n}$, where $\onevec{n} = (1,1,\dots, 1)^{T} \in \Rspace^{n}$. In the following sections, we slightly abuse the  notation and identify a graph $G$ with a particular choice of measure network representation $(V,p,W)$.

A measure network $G = (V, p, W)$ may be equipped with additional information on its nodes, namely, the \emph{node attributes}. 
That is, we associate each node $x \in V$ with an attribute $a$ in some attribute space---a metric space denoted as $(A, d_{A})$.
Possible node attributes include labels on the nodes or information derived from the data domain from which $G$ arises.

\subsection{Wasserstein and Gromov-Wasserstein  Distance}
\label{sec:fgw}

Let $G_{1}=(V_1, p_{1}, W_{1})$ and $G_{2}=(V_2, p_{2}, W_{2})$ be a pair of measure networks with $n_{1}$ and $n_2$ nodes, respectively. 
Let $[n]$ denote the set $\{1,2,\dots,n\}$ and suppose that $V_1=\{x_i\}_{i \in [n_1]}$ and $V_2=\{y_j\}_{j \in [n_2]}$. 
A \emph{coupling} between probability measures $p_{1}$ and $p_{2}$ is a joint probability measure on $V_1 \times V_2$ whose marginals agree with $p_1$ and $p_2$. 
That is, a coupling is represented as an $n_{1}\times n_{2}$ non-negative matrix $C$ such that $C\onevec{n_{2}} = p_{1}$ and $C^{T}\onevec{n_{1}} = p_{2}$. 
The set of all such couplings is denoted as $\Ccal$, that is, 
\begin{align}
\label{eq:couplings}
\Ccal = \Ccal(p_1, p_2) = \{C \in \Rspace_{+}^{n_1 \times n_2} \mid C\onevec{n_{2}} = p_{1}, C^{T}\onevec{n_{1}} = p_{2}\}. 
\end{align}

\para{Wasserstein distance.}
Classical optimal transport theory compares probability measures in terms of the Wasserstein distance. 
Given a pair of measure networks $G_1=(V_1, p_1, W_1)$ and $G_2=(V_2,  p_2, W_2)$, where nodes $x_i \in V_1$ and $y_j \in V_2$ are equipped with attributes $a_i$ and $b_j$ within the same attribute space, we define their \emph{$q$-th Wasserstein distance} based on distances between node attributes to be  
\begin{align}
d^{W}_q(G_{1},G_{2}) = \min_{C\in\Ccal} \left(\sum_{i,j} d_A(a_i,b_j)^q C_{i,j}\right)^{1/q}.
\label{eq:w-distance}
\end{align}
We refer to $d_A(a_i, b_j)$ as the \emph{attribute distance} between nodes $x_i \in V_1$ and $y_j \in V_2$. 
The Wasserstein distance aims to minimize the weighted sum of attribute distance between matched nodes. 
The minimizers in~\eqref{eq:w-distance} are referred to as \emph{optimal couplings}. 

\para{GW distance} was introduced by M\'{e}moli as a way to compare metric measure spaces~\cite{Memoli2007,Memoli2011}.  Chowdhury and M\'{e}moli~\cite{ChowdhuryMemoli2019} showed that a generalized GW distance is a metric on the space of \emph{measure networks}. 
The key idea behind the GW distance is to find a \emph{probabilistic matching} between a pair of measure networks by searching the convex set of couplings of the probability measures defined on the networks. 
Following~\cite{ChowdhuryMemoli2019}, the \emph{$q$-th GW distance} between two measure networks is defined as 
\begin{align}
d^{GW}_{q}(G_{1},G_{2}) = \frac{1}{2} \min_{C\in\Ccal} \left(\sum_{ i,j,k,l}  |W_1(i,k) - W_{2}(j,l)|^q C_{i,j}C_{k,l} \right)^{1/q}.
\label{eq:gw}
\end{align}
The term $|W_1(i,k) - W_{2}(j,l)|$ is considered as the \emph{distortion} of matching pairs of nodes $(x_i, x_k)$ in $G_1$ with $(y_j,y_l)$ in $G_2$. 

\para{FGW distance.} 
Vayer {\etal} introduced the fused Gromov-Wasserstein (FGW) distance between attributed graphs and other structured objects~\cite{VayerCourtyTavenard2019,VayerChapelFlamary2020}. 
We describe their framework in the setting of measure networks.  
The FGW distance is a trade-off between the Wasserstein distance in~\eqref{eq:w-distance} and the GW distance in~\eqref{eq:gw}. 
For $q \in [1, \infty)$ and a trade-off parameter $\alpha \in [0,1]$, the \emph{FGW distance}  between attributed measure networks $G_{1}$ and $G_{2}$ is defined (following \cite{VayerCourtyTavenard2019}) as  
\begin{align}
& d^{FGW}_q(G_{1},G_{2}) = \nonumber\\
& \min_{C\in\Ccal}  \sum_{ i,j,k,l} [(1-\alpha) d_A(a_i,b_j)^q 
+ \alpha |W_1(i,k) - W_{2}(j,l))|^q] C_{i,j}C_{k,l}.
\label{eq:fgw}
\end{align}
Here, $C$ is considered as a soft assignment matrix, and $\alpha$ gives a trade-off between labels and structures. 
As shown in~\autoref{sec:method}, \eqref{eq:fgw} plays an important role in encoding both intrinsic and extrinsic information for merge tree comparisons. 

The FGW distance enjoys a number of desirable properties (see~\cite{VayerCourtyTavenard2019} and its supplementary material, as well as~\cite{VayerChapelFlamary2020}). 
Specifically, it interpolates between the Wasserstein distance on the labels and GW distances on the structures:  
\begin{theorem}\cite[Theorem 3.1]{VayerCourtyTavenard2019}
As $\alpha \to 0$, the FGW distance recovers the Wasserstein distance, 
\begin{align}
\lim_{\alpha \to 0} d^{FGW}_q = (d^{W}_q)^q. 
\end{align}
As $\alpha \to 1$, the FGW distance recovers the GW distance  (ignoring the constant factor in~\eqref{eq:gw}), 
\begin{align}
\lim_{\alpha \to 1} d^{FGW}_q = (d^{GW}_q)^q. 
\end{align}
\end{theorem} 
Furthermore, $d^{FGW}_q$ defines a metric for $q = 1$ and a semimetric for $q \geq 2$ (\ie, the triangular inequality is relaxed by a factor $2^{q-1}$)~\cite[Theorem 3.2]{VayerCourtyTavenard2019}. 

For the remainder of the paper, we work with $d^{FGW}_q$ for $q=2$. 
For easy reference, we have
\begin{align}
& d^{FGW}_2(G_{1},G_{2}) = \nonumber \\
& \min_{C\in\Ccal}  \sum_{ i,j,k,l} [(1-\alpha) d_A(a_i,b_j)^2 
+ \alpha |W_1(i,k) - W_{2}(j,l))|^2] C_{i,j}C_{k,l}.
\label{eq:fgw2}
\end{align}
The choice of $q=2$ is justified for computational reasons: given two measure networks with $n_1$ and $n_2$ nodes, respectively, we can simplify the computation of the tensor product involved in the evaluation of the GW loss from 
$\mathcal{O}(n_1^2 n_2^2)$ to $\mathcal{O}(n_1 n_2^2 + n_1^2 n_2)$ when considering $q=2$~\cite{PeyreCuturiSolomon2016}.

\subsection{Partial Wasserstein and Partial GW Distances}
\label{sec:pgw}

Our final ingredient comes from partial optimal transport (see, \eg, \cite{figalli2010optimal,benamou2015iterative,chizat2018scaling}). 
We review the framework of Chapel \etal~\cite{ChapelAlayaGasso2020} that studies partial Wasserstein and partial GW distances.  
Notations are simplified in our setting of measure networks. 
Partial optimal transport is appropriate in the setting of feature tracking, when we need to account for mass changes due to the appearances and disappearances of features.   

\para{Partial Wasserstein distance.}
Partial optimal transport focuses on transporting a fraction $0 \leq m \leq 1$ of the mass as cheaply as possible~\cite{ChapelAlayaGasso2020}. 
The set of admissible couplings is defined to be 
\begin{align}
\label{eq:pw-couplings}
&\Ccal_m = \Ccal_m(p_1, p_2) \nonumber \\
&= \{C \in \Rspace_{+}^{n_1 \times n_2} \mid C\onevec{n_{2}} \leq p_{1}, C^{T}\onevec{n_{1}} \leq p_{2}, \mathbf{1}_{n_{1}}^T C \onevec{n_{2}} = m \},
\end{align}
and the \emph{partial $q$-Wasserstein distance} is defined as 
\begin{align}
d^{pW}_q(G_{1},G_{2}) = \min_{C\in\Ccal_m} \left(\sum_{i,j} d_A(a_i,b_j)^q C_{i,j}\right)^{1/q}.
\label{eq:pw-distance}
\end{align}
A main difference between partial Wasserstein distance and Wasserstein distance is that we replace the equalities in~\eqref{eq:couplings} with inequalities in~\eqref{eq:pw-couplings} to account for ``partial mass transport''.

\para{Partial GW distance.}
In a similar fashion, given the set of admissible couplings $\Ccal_m$, the \emph{partial q-GW distance} is defined as
\begin{align}
d^{pGW}_{q}(G_{1},G_{2}) = \frac{1}{2} \min_{C\in\Ccal_m} \left(\sum_{ i,j,k,l} |W_1(i,k) - W_{2}(j,l)|^q C_{i,j}C_{k,l} \right)^{1/q}.
\label{eq:pgw}
\end{align}

\section{Method}
\label{sec:method} 

We now describe our novel framework that performs feature tracking with partial optimal transport. 
We first introduce a new, partial Fused Gromov-Wasserstein (pFGW) distance between a pair of measure networks (\autoref{sec:pfgw}). 
We then model and compare merge trees as measure networks (\autoref{sec:model}).
The pFGW distance gives rise to a partial matching between topological features (\ie, critical points) in merge trees, thus enabling flexible topology tracking for time-varying data (\autoref{sec:tracking}). 

\subsection{Partial Fused Gromov-Wasserstein Distance} 
\label{sec:pfgw} 
For topology-based feature tracking, oftentimes features (\ie, critical points) will appear and disappear in time-varying data. Features that appear at time $t$ do not need to be matched with features at time $t-1$; similarly, features that disappear at time $t$ do not need to be matched with features at time $t+1$. 
Therefore, we need to introduce a partial Fused Gromov-Wasserstein (pFGW) distance for feature tracking to handle the appearances and disappearances of multiple features across time.  

The \emph{pFGW distance} is defined based on the set of admissible couplings $\Ccal_m$ in \eqref{eq:pw-couplings} and the FGW distance in \eqref{eq:fgw}. 
Given a pair of measure networks $G_1$ and $G_2$, 
formally, we have 
\begin{align}
& d^{pFGW}_q(G_{1},G_{2}) = \nonumber\\
& \min_{C\in\Ccal_m}  \sum_{ i,j,k,l} [(1-\alpha) d_A(a_i,b_j)^q 
+ \alpha |W_1(i,k) - W_{2}(j,l))|^q] C_{i,j}C_{k,l}.
\label{eq:pfgw}
\end{align}
Notice that the newly defined pFGW distance is not too different from the FGW distance, except that it is more flexible by allowing a $m$ fraction of the total mass to be transported.  
In practice, we set $q=2$ and work with $d^{pFGW}_2$.

We remark that a related distance was recently introduced in~\cite{thual2022aligning} and applied to brain anatomy alignment. The difference between the two distances is that \cite{thual2022aligning} employs a different notion of partial optimal transport (rather, \emph{unbalanced} optimal transport), where the coupling set is expanded to all joint probability measures and disagreement of marginals is penalized by Kullback-Liebler (KL) divergence. In~\cite{thual2022aligning}, instead of choosing the amount of mass to be preserved, one must tune the relative weight of the KL regularization term.

\para{Computing pFGW distance.}
Computing the pFGW distance is a slight modification of the FGW computation in~\cite{VayerChapelFlamary2020} with ingredients of the Frank-Wolfe optimization algorithm~\cite{FrankWolfe1956} for partial GW computation~\cite{ChapelAlayaGasso2020}. 
On a high-level, computing the partial Wasserstein and the partial GW distances relies on adding dummy nodes in the transportation plan and allowing such dummy nodes to ``absorb'' a fraction of the mass during transportation.  
With these dummy nodes added onto the marginals, the Frank-Wolfe algorithm then solves an iterative first-order optimization for constrained convex optimization. 
Our implementation is based on a minor modification of the code for the FGW framework in~\cite{VayerCourtyTavenard2019} (\url{https://github.com/tvayer/FGW}) with components from the partial optimal transport solvers, part of the open-source Python library for optimal transport~\cite{flamary2021pot} 
(\url{https://pythonot.github.io/gen_modules/ot.partial.html}). 

\subsection{Modeling Merge Trees as Measure Networks}
\label{sec:model}

Unless otherwise specified, we represent a merge tree $T$ as an attributed measure network $(V, p, W)$ for the remainder of this paper, where the attributes, weight matrix $W$, and probability measure $p$ are defined below. 

Given a merge tree $T = (V, p, W)$, information that is typically topological and intrinsic to a merge tree, such as tree distances, may be encoded via the weight matrix $W$ and the probability measure $p$ (\autoref{sec:intrinsic}). 
Information that is extrinsic to a merge tree may be encoded via the \emph{node labels} $(A, d_A)$. 
Extrinsic information is typically geometrical or statistical, and arises from the data domain, such as the coordinates of the critical points of $f: \Mspace \to \Rspace$ (that give rise to the merge tree), function values $f$ restricted to the set of nodes $V$, and prior knowledge (such as labels) associated with nodes in a measure network. 

We discuss various strategies that encode extrinsic and intrinsic information for merge tree comparisons. 
The key takeaway is that the pFGW distance we build upon  provides a flexible framework that encodes geometric and topological information for comparative analysis of merge trees. 

\subsubsection{Encoding Intrinsic Information}
\label{sec:intrinsic}
A merge tree $T$ is represented using a triple $(V, p, W)$. 
Information intrinsic to $T$ may be encoded via $p$ and $W$ as we now describe.

\para{Encoding edge information.}  
Recall that a merge tree $T$ is a tree equipped with a function $f: V \to \Rspace$ defined on its nodes $V$. 
To encode the information of $f$, we explore a \emph{shortest path strategy}. 
Recall that each node $x$ in $T$ is associated with a scalar value $f(x)$. For $x,x' \in V$, we define $W(x,x')$ as follows: we associate the weight $W(x,x') = |f(x) - f(x')|$ with each pair of adjacent nodes; for nonadjacent nodes, $W(x,x')$ is the sum of the edge weights along the unique shortest path in $T$  from $x$ to $x'$. 
By construction, the shortest path between two nodes goes through their lowest common ancestor in $T$. That is, an \emph{ancestor} of a node $x$ in $T$ is any node $v$ such that there exists a path from $x$ to $v$ where $f$-values are non-decreasing along the path. The \emph{lowest common ancestor} of two nodes $x,x'$, denoted $\lca(x,x')$, is the common ancestor of $x$ and $x'$ with the lowest $f$-value.

We explore an additional strategy by encoding the function values  of the lowest common ancestors among pairs of nodes, referred to as the \emph{lowest common ancestor strategy}. Using this strategy, we define $W(x,x') = f(\lca(x,x'))$ for $x, x' \in V$. For a given ordering of vertices, $W$ is also known as the \emph{induced ultra matrix} of a merge tree~\cite{GasparovicMunchOudot2019}. 

\para{Encoding node information.}
Without prior knowledge, we may define $p$ as a uniform measure, \ie, $p = \frac{1}{|V|} \onevec{|V|}$. 
This \emph{uniform strategy} means that all nodes in the merge trees are considered to be equally important during merge tree comparison and matching.  

On the other hand, $p$ could be made more general by giving higher weights to nodes deemed more important by an application. 
For example, we may assign each node $x \in V$ an importance value that is proportional to the functional difference to its \emph{parent node}, $\parent(x)$, which is the unique neighbor $x'$ of $x$ in $T$ with $f(x') \geq f(x)$. That is, we set $p(x) \propto (f(\parent(x)) - f(x))$. Such assignment is referred to as the \emph{parent strategy}.

\subsubsection{Encoding Extrinsic Information}
\label{sec:extrinsic}

Extrinsic information that typically arises from the geometry of the data domain may be encoded via the attribute space $(A, d_A)$ and attribute distance $d_A$ in~\eqref{eq:fgw}. 
For a node in the merge tree, the assigned attribute may be a high-dimensional vector or a categorical label.
Given a pair of merge trees $T_1 = (V_1, p_1, W_1)$ and $T_2 = (V_2,  p_2, W_2)$, nodes $x_i \in V_1$ and $y_j \in V_2$ are equipped with attributes $a_i$ and $b_j$ from the same attribute space $(A, d_A)$.  

These attributes may be coordinates associated with critical points in the data domain. 
Specifically, assume attribute $a_i = (x^1_i, x^2_i) \in V_1$ is associated with a critical point of $f_1: \Mspace \to \Rspace$ in the data domain $\Mspace$ with coordinates $(x^1_i, x^2_i)$ (assuming $\Mspace \subset \R^2$), whereas attribute  $b_j = (y^1_j, y^2_j) \in V_2$ corresponds to a critical point of $f_2: \Mspace \to \Rspace$ with coordinates $(y^1_j, y^2_j)$. 
 We define $d_A$ to be the Euclidean distance between $a_i$ and $b_j$, $d_A(a_i, b_j) = \sqrt{(x^1_i - y^1_j)^2 + (x^2_i - y^2_j)^2}$. 
This definition is referred to as the \emph{coordinates strategy}. 
This strategy is a natural choice because a core method for critical point tracking is often based on their Euclidean distance proximity. 

Another useful node attribute is the type (category) of critical points (\eg, local maximum, local minimum, and saddle). Assuming attributes $a_i$ and $b_j$ capture the categories of critical points $x_i$ from $f_1$ and $y_j$ from $f_2$ respectively, we may define the \emph{category distance} between $a_i$ and $b_j$ as 
\begin{equation}
d_A(a_i, b_j) = \left\{    
\begin{array}{rcl}
0 & & a_{i} = b_{j}; \\
1 & & a_{i} \neq b_{j}.
\end{array} \right.
\label{eq:penalty}
\end{equation}
That is, the distance between categories is $1$ if the categories do not match and $0$ if they do. 
This is referred to as the \emph{category strategy}. 
In practice, we combine the above two distances to form the attribute distance, referred to as the \emph{combined strategy.}

 \subsubsection{Simple Examples}
 \label{sec:example}
In~\autoref{fig:pfgw-gaussian}, we show a simple example of using pFGW distance for critical point matching between a pair of merge trees $T_1$ and $T_2$.
These merge trees arise from slightly different mixtures of Gaussian functions $f_1$ and $f_2$ in 2D; see (a) and (c), respectively. 
As shown in (a), $T_1$ and $T_2$ are structurally similar: $T_1$ contains $10$ critical points, and $T_2$ has $8$ critical points with a pair of critical points removed (see the region enclosed by the red box). 
Here, we apply a \emph{uniform strategy} to $p$.  
We set $m=0.8$, because 2 of 10 nodes in $T_1$ no longer exist in $T_2$. 
After computing the pFGW distance, the $10 \times 8$ coupling matrix $C$ is shown in (d) and visualized in (b).  
An entry $C(i, j)$ in the coupling matrix indicates the probability of a node $i \in T_1$ being matched to node $j \in T_2$. 
In particular, rows $C(2, \cdot)$ and $C(3, \cdot)$ (in a red box) are both zero, indicating that no partners in $T_2$ are matched with nodes $2$ and $3$ in $T_1$.  
Furthermore, nodes in $T_2$ are colored by its most probable partner in $T_1$, which aligns well with our intuition that all nodes with the same property should be matched to each other. 
In this example, each node in $T_1$ has a unique partner in $T_2$; however, in practice, a node may be coupled with multiple nodes with nonzero probabilities, as shown in the next example.  

\begin{figure}[!ht]
    \centering
    \includegraphics[width=0.98\columnwidth]{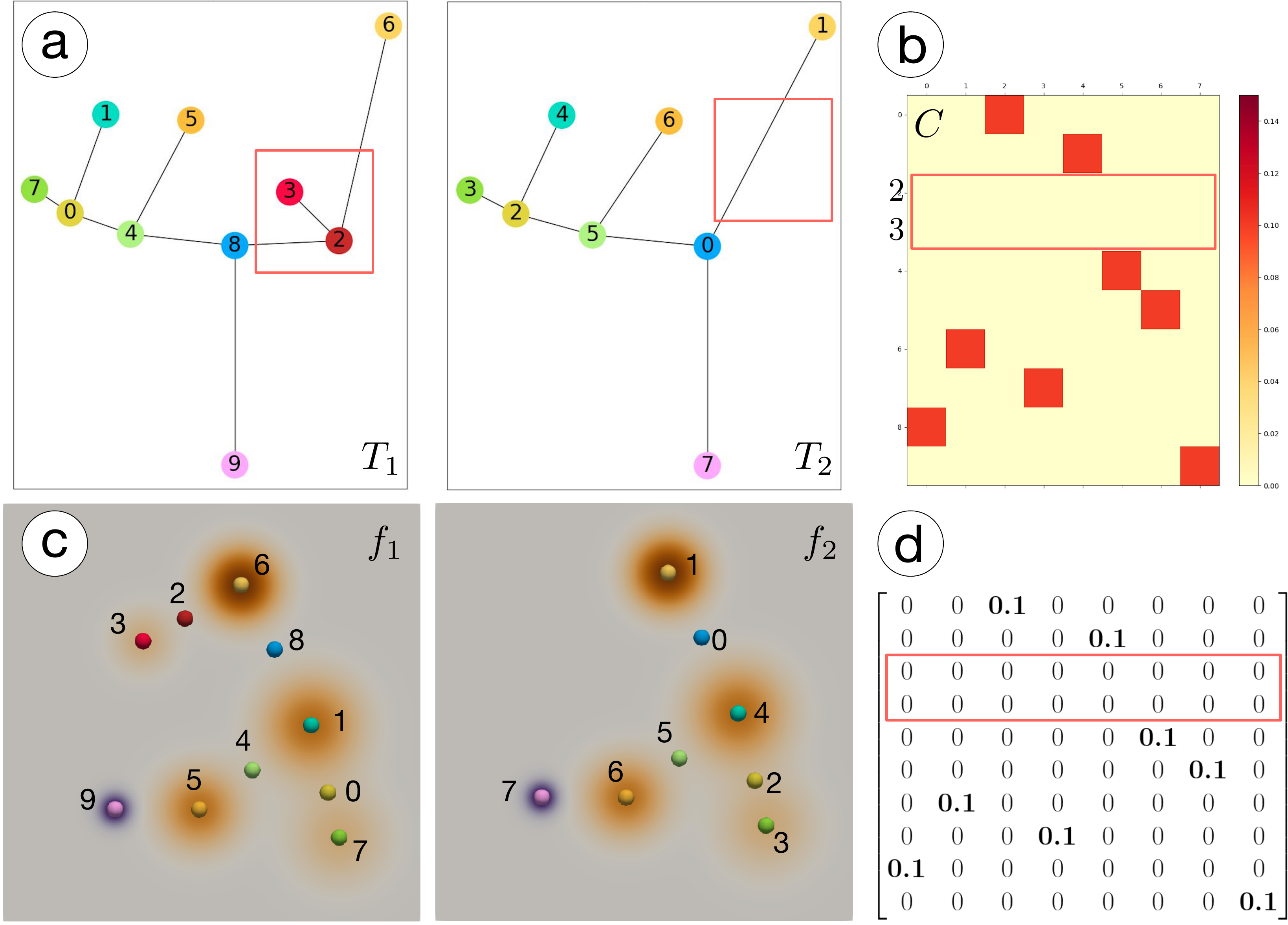}
    \vspace{-2mm}
    \caption{Partial optimal matching using pFGW, $m=0.8$. (a) Merge trees that arise from mixtures of Gaussian functions in (c). The coupling matrix (d) is visualized with a heat map in (b).}
    \label{fig:pfgw-gaussian}
\end{figure}

We provide another example in~\autoref{fig:pfgw-gaussian2} to demonstrate probabilistic matching with our framework. 
As shown in~(c), $f_1$ is a mixture of four positive and one negative Gaussian functions. 
In $f_2$, a positive Gaussian function on top is split into two Gaussian functions, resulting in two local maxima and one saddle point. 
Merge trees $T_1$ and $T_2$ in (a) describe the topology of scalar fields $f_1$ and $f_2$, respectively. 
Notice that the topological change in $T_2$ (enclosed by a red box) highlights the feature splitting event in $f_2$. 
Rather than enforcing a one-to-one correspondence between the critical points, our pFGW framework allows probabilistic matching among them. 
As shown in (b) and (d), the coupling matrix $C$ contains multiple rows and columns with more than one nonzero entry.
For example, the row $C(3, \cdot)$ (red box) has two nonzero entries, namely $0.025$ at $C(3,1)$ and $0.01$ at $C(3,4)$, see (d), 
which indicates that node $3$ in $T_1$ can be matched to both node $1$ and $4$ in $T_2$ with varying probabilities. 
Such a matching is probable due to the feature splitting event. 
As node $4$ in $T_2$ is closer to node $3$ in $T_1$ (than node $1$ in $T_2$), $C(3,4)$ has a higher coupling probability than $C(3,1)$. 

\begin{figure}[!ht]
    \centering
    \includegraphics[width=0.98\columnwidth]{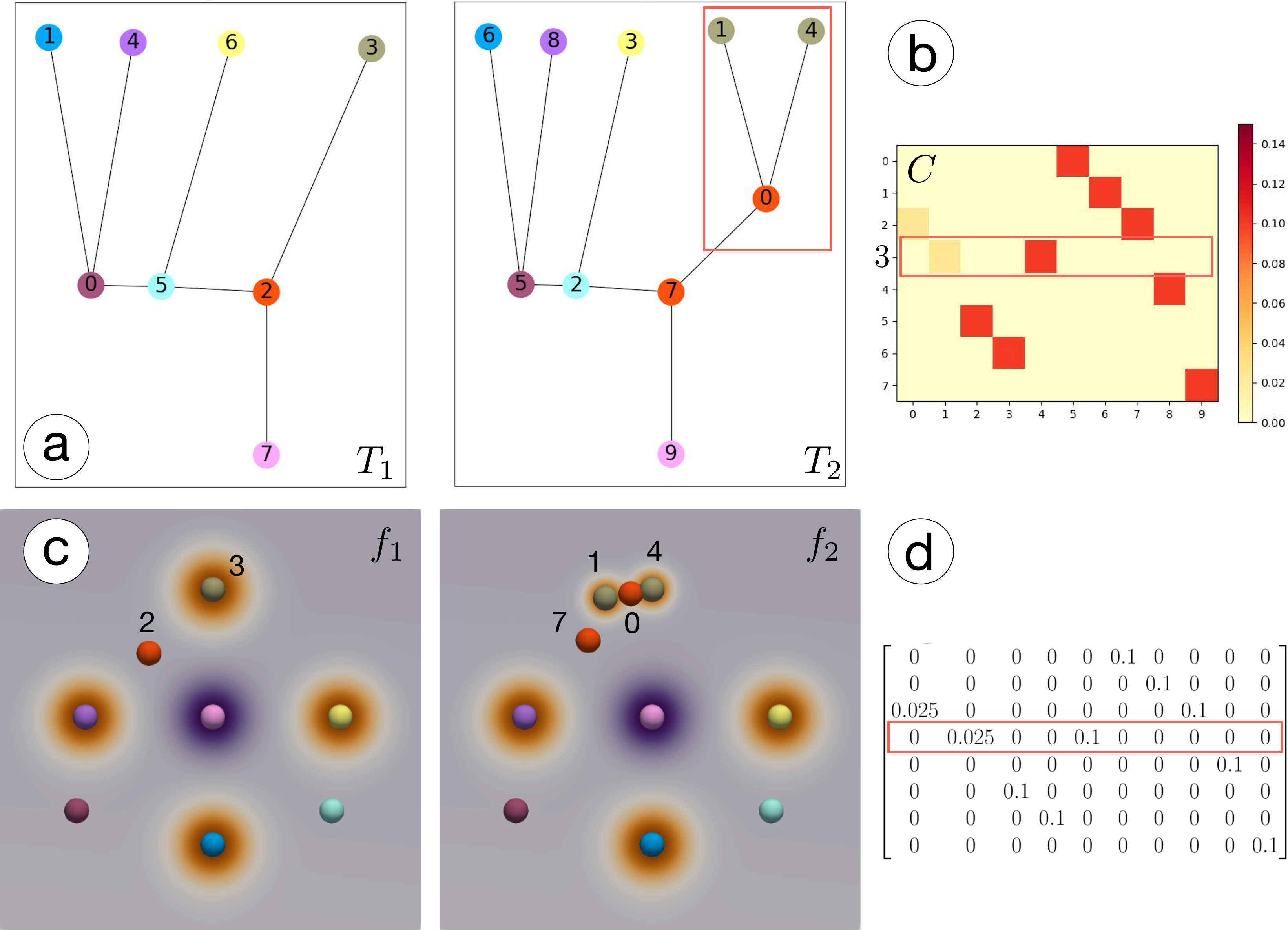}
    \caption{Partial optimal matching using pFGW, $m = 0.85$. (a) Merge trees that arise from mixtures of Gaussian functions in (c). The coupling matrix (d) is visualized with a heat map in (b).}
    \label{fig:pfgw-gaussian2}
\end{figure}

\subsection{Flexible Topology Tracking}
\label{sec:tracking}

By modeling merge trees as measure networks (\autoref{sec:model}) and introducing a new pFGW distance based on partial optimal transport (\autoref{sec:pfgw}), we are ready to describe our topology tracking framework in~\autoref{sec:framework} and discuss its flexibility in~\autoref{sec:flexibility}.

\subsubsection{Tracking Framework}
\label{sec:framework}
Our topology tracking framework consists of three steps. 

\para{1. Feature detection.}
First, we compute a merge tree for each time step. We use the algorithm implemented in TTK~\cite{TiernyFavelierLevine2018, gueunet2017task, lukasczyk2020localized}. 
Each merge tree contains local minima, saddles, and a global maximum (assuming there is a unique global maximum). 
When the data is noisy, we apply persistent simplification~\cite{EdelsbrunnerLetscherZomorodian2002} to remove pairs of critical points with low persistence, in order to retain significant features in the domain for tracking purposes.  

\para{2. Feature matching.}
Second, we utilize our pFGW framework for feature matching across adjacent time steps. 
Let $T_1$ and $T_2$ be two merge trees computed at time steps $t$ and $t+1$, respectively. We then model them as measure networks $T_1=(V_1, p_1, W_1), T_2=(V_2, p_2, W_2)$ and apply the pFGW framework described in~\autoref{sec:pfgw} to match critical points from $T_1$ with $T_2$. 

We utilize a conservative \emph{bijective matching strategy}. 
Based on the optimal coupling $C$, a node $x \in V_1$ may be coupled (matched) with multiple nodes in $V_2$. We will choose $x' \in V_2$, which has the highest matching probability with $x$ (referred to as the most probable partner). 
Similarly, for $x' \in V_2$, we will choose its most probable partner $x'' \in V_1$. 
If $x = x''$, then $x$ and $x'$ are matched to form a trajectory. 

\para{3. Trajectory extraction.}
Trajectories are constructed by connecting successively matched critical points. 
For any two adjacent time steps $t$ and $t+1$, if a node $x$ at time $t$ is matched with a node $x'$ at time $t+1$, then a segment is constructed connecting $x$ and $x'$ in the spacetime domain. 
If a node $x$ at time $t$ is ignored (\ie, matched to the dummy node) during the partial optimal transport, then the current trajectory terminates. 
If a node $x'$ at time $t+1$ is ignored during the partial optimal transport, it is considered as a new feature, and a new trajectory begins. 

\subsubsection{A Discussion on Flexibility}
\label{sec:flexibility}
Modeling a merge tree $T$ as a measure network $T = (V, p, W)$ and its associated pFGW distance offers great flexibility in the comparative analysis of merge trees. 
The flexibility is reflected via a number of parameters. 

First, parameters $W$ and $p$ allow various strategies for encoding intrinsic and extrinsic information of a merge tree, including the \emph{shortest path} and \emph{lowest common ancestor} strategies for encoding edge information; \emph{uniform} and \emph{parent} strategies for encoding node information; \emph{coordinates}, \emph{category}, and their \emph{combined} strategy for encoding geometric information from the data domain. 

Second, parameter $\alpha$ from~\eqref{eq:pfgw} strikes a balance in considering intrinsic information (via the GW distance) and extrinsic information (via the Wasserstein distance) for merge tree comparisons. 

Third, parameter $m$ from~\eqref{eq:pfgw} allows partial mass transport to accommodate the appearances and disappearances of features.

\section{A New Stability Result} 
\label{sec:stability-overview}

We now state a new theoretical stability result involving the GW distance, which shows that a small change in the function data produces a small change in merge tree representations, as measured by the GW distance; see the supplementary material  for a detailed proof and some experimental validation of \autoref{thm:stability-main}. 

Let $X$ be a finite, connected geometric simplicial complex with vertex set $V$. 
Let $f:X \to \R$ be a function obtained by starting with a function $f: V \to \R$ on the vertex set and extending linearly over higher dimensional simplices. 
Let $p$ be a probability distribution over the vertex set $V$. We will assume that $p$ is \emph{balanced}, in the sense that for any $u,v,w \in V$, we have $p(u) \cdot p(v) \leq p(w)$; this property holds for the uniform distribution, for example. We then define the measure network representation of merge tree of $f$ to be $G_f=(V,p,W_f)$, with $W_f$ defined based on the least common ancestor strategy. We also define a family of weighted norms on the space of functions $f: V \to \R$ by 
\[
\|f\|_{L^q(p)} := \left(\sum_{v \in V} |f(v)|^q p(v)\right)^{1/q}.
\]
We can now state our theorem.
\begin{theorem}
\label{thm:stability-main}
Let $f,g:X \to \R$ be functions defined as above and let $p$ be a balanced probability distribution. Then
\begin{equation*}
d^{GW}_q(G_f,G_g) \leq \frac{1}{2} |V|^{2/q} \|f - g\|_{L^q(p)}.
\end{equation*}
\end{theorem}

We also show in the supplementary material that the Lipschitz constant $\frac{1}{2}|V|^{2/q}$ is asymptotically tight for general probability measures. When the measure is uniform, the constant can be improved to $\frac{1}{2}|V|^{1/q}$. Finally, we have the following corollary, which treats the shortest path strategy for encoding a merge tree as a measure network.

\begin{corollary}\label{cor:stability}
Let $f,g:X \to \R$ be functions defined as above and let $p$ be a balanced probability distribution. Let $G_f$ (respectively, $G_g$) denote the representation of the merge tree $T_f$ (respectively, $T_g$) defined by the shortest path strategy. Then
\begin{equation*}
d^{GW}_q(G_f,G_g) \leq \left(|V|^{2/q} + 2\right) \|f - g\|_{L^q(p)}.
\end{equation*}
\end{corollary}

We now briefly comment on the structure of this stability result, and, in particular, on its dependency on $|V|$. There are several metrics on the space of merge trees, or more generally, Reeb graphs, which enjoy stability results of the same form, but which are apparently stronger in that they do not depend on the combinatorics of the domain (i.e., such that the Lipschitz constant is absolute). This is the case, for example, for the functional distortion distance~\cite{bauer2014measuring}, interleaving distance ~\cite{morozov2013interleaving,de2016categorified}, merge tree matching distance ~\cite{Bollen2023}, and the Reeb graph edit distance~\cite{Bauer2018b, DiFabio2016}. However, these metrics are all $L^\infty$-type distances, and the most appropriate comparison to our result would involve taking $q \rightarrow \infty$, in which case the dependency on $|V|$ vanishes. This behavior is comparable to the recent $p$-Wasserstein stability result of~\cite{skraba2020wasserstein} in the context of $L^p$-type distances between persistence diagrams~\cite{cohen2007stability,chazal2009gromov}. 
\section{Experiments}
\label{sec:results}

We demonstrate the utility of our framework with five 2D datasets and two 3D datasets.  
For each dataset, we also compare against two state-of-the-art approaches.  
In particular, we demonstrate the strengths of our framework using two complex datasets---{\CL} and {\VF}---in tracking a large number of features.

\subsection{Datasets Overview}
\label{sec:datasets}

The {\HC} dataset is a simulation of a 2D flow generated by a heated cylinder using the Boussinesq approximation~\cite{gerrisflowsolver,GuntherGrossTheisel2017}. 
The simulation was done with a Gerris flow solver. It shows a time-varying turbulent plume containing numerous small vortices that, in part, rotate around each other. 
We generate a set of merge trees from the magnitude of the velocity fields based on $31$ time steps ($600$-$630$ from the original $2000$ time steps).
These time steps describe the evolution of small vortices.

The {\UCF} dataset is a 2D unsteady cylinder flow. 
This synthetic vector field was created by Jung, Tel, and Ziemniak~\cite{JungTelZiemniak93} and serves as a basic model of a von-K\'arman vortex street generation. 
We use the first 499 time steps in the dataset, and use merge trees computed for the velocity magnitude field that primarily capture the behavior of local maxima, saddles, and a global minimum. 
Both {\HC} and {\UCF} datasets are available via the Computer Graphics Laboratory~\cite{cgl}. 

The  {\VS} dataset is the classic 2D von K\'arman vortex street dataset coming from the simulation of a viscous 2D flow around a cylinder. 
It contains vortices moving with almost constant speed to the right, except directly in the wake of the obstacle, where they accelerate. 
We model vorticity magnitude as scalar fields, and track the evolution of local maxima over time. 

The {\IF} dataset comes from the 2008 IEEE Visualization Design Contest~\cite{sciVis2008}. 
It simulates the propagation of an ionization front instability.
The simulation is done with 3D radiation hydrodynamical calculations of ionization front instabilities in which multi-frequency radiative transfer is coupled to the primordial chemistry of eight species~\cite{WhalenNorman2008}. 
We use the density to generate merge trees from the 2D slices near the center of the simulation volume for 123 time steps, which correspond to steps 11-133 from the original 200 time steps. 
These time steps show the density over time as the instability progresses toward the right.

The {\CL} dataset shows the cloud optical thickness retrieved via the Daytime Cloud Optical and Microphysical Properties Algorithm (DCOMP) by Walther and Heidinger ~\cite{Walther2012daytime}, processed by Chatterjee \etal~\cite{ChatterjeeSchnittBigalke2023}. This 2D dataset has been used previously for cloud tracking~\cite{Nouri2019}. We focus on the data sampled every 10 minutes from 10:50 to 16:50 on Feb 2, 2020 within the region of $10.82^{\circ}$N - $15.88^{\circ}$N, $49.19^{\circ}$W - $42.51^{\circ}$W. 
We use this dataset to demonstrate the utility of our pFGW framework on tracking a large number of features.

The {\IS} dataset is a collection of 3D volumes simulating the wind velocity magnitude of the Isabel hurricane. 
We use this dataset to demonstrate the ability of our method to track features in 3D scientific datasets. 
We use $12$ time steps that depict the key events of the hurricane (formation, drift, and landfall): time steps $2$ to $5$, $30$ to $33$, and $45$ to $48$.
This 3D dataset is acquired from the Climate Data Gateway
at NCAR~\cite{NCAR}.

The {\VF} dataset comes from the 2016 IEEE Scientific Visualization contest~\cite{sciVis2016}. This ensemble dataset simulates transient fluid flows coming from the mixing process of salts solving into a cylinder of water. During this process, the structures of increased salt concentration values are called the viscous fingers. Previous works~\cite{LukasczykWeberMaciejewski2017, lukasczyk2019dynamic} have used the superlevel set components of the concentration field for tracking the viscous fingers. Here, we use the  local maxima of the concentration field to track the viscous fingers. We select the data from the  first run of the ensemble (with a smoothing length of $0.44$), which contains $120$ time steps.
 
\subsection{Heated Cylinder Dataset}
\label{sec:heated-flow}
We first use the {\HC} dataset to demonstrate in detail our parameter tuning process in~\autoref{sec:tuning}. 
We then showcase the tracking results based on partial optimal transport in~\autoref{sec:hc-tracking-results}. 
Finally, we compare against previous approaches in~\autoref{sec:hc-compare}. 

\subsubsection{Parameter Tuning}
\label{sec:tuning}

\para{Evaluation metrics.}
To evaluate the quality of the extracted trajectories, we aim to reduce two types of artifacts during parameter tuning: \emph{oversegmentations} where a single trajectory is unnecessarily segmented into subtrajectories; and \emph{mismatches} between critical points that appear as zigzag patterns connecting (often faraway) critical points from adjacent time steps.

We introduce two metrics to evaluate these artifacts quantitatively: first, the \emph{number of trajectories}, denoted as $N$; and second, the maximum Euclidean distance between matched critical points across time (referred to as the \emph{maximum matched distance} for simplicity), denoted as $L$.  

There are two types of parameters in our framework: the preprocessing parameter $\epsilon$ that is used to de-noise the input data; and the in-processing parameters $W$, $p$, $\alpha$, and $m$ for feature tracking.  

\para{Preprocessing parameter tuning.} 
Persistence simplification is considered a preprocessing step for data de-noising. 
Let $\epsilon \in [0,1]$ denote the persistence simplification parameter.
Let $R$ denote the range of a given scalar field. 
Using persistence simplification, critical points with persistence less than $\epsilon \cdot R$ are removed from the domain.  
$\epsilon$ is typically chosen based on the shape of a  \emph{persistence graph}, where a plateau in a persistence graph indicates a stable range of scales to separate features from noise. Such a strategy has been used previously in simplifying scientific data (\eg,~\cite{GerberBremerPascucci2010,BremerMaljovecSaha2015}).  
For {\HC}, we use $\epsilon=6\%$, which is slightly left of the first observable plateau in the persistence graph, as we try to maintain a slightly larger number of features; see~\autoref{fig:HC-persistence}.

 \begin{figure}[!ht]
    \centering
    \includegraphics[width=0.98\columnwidth]{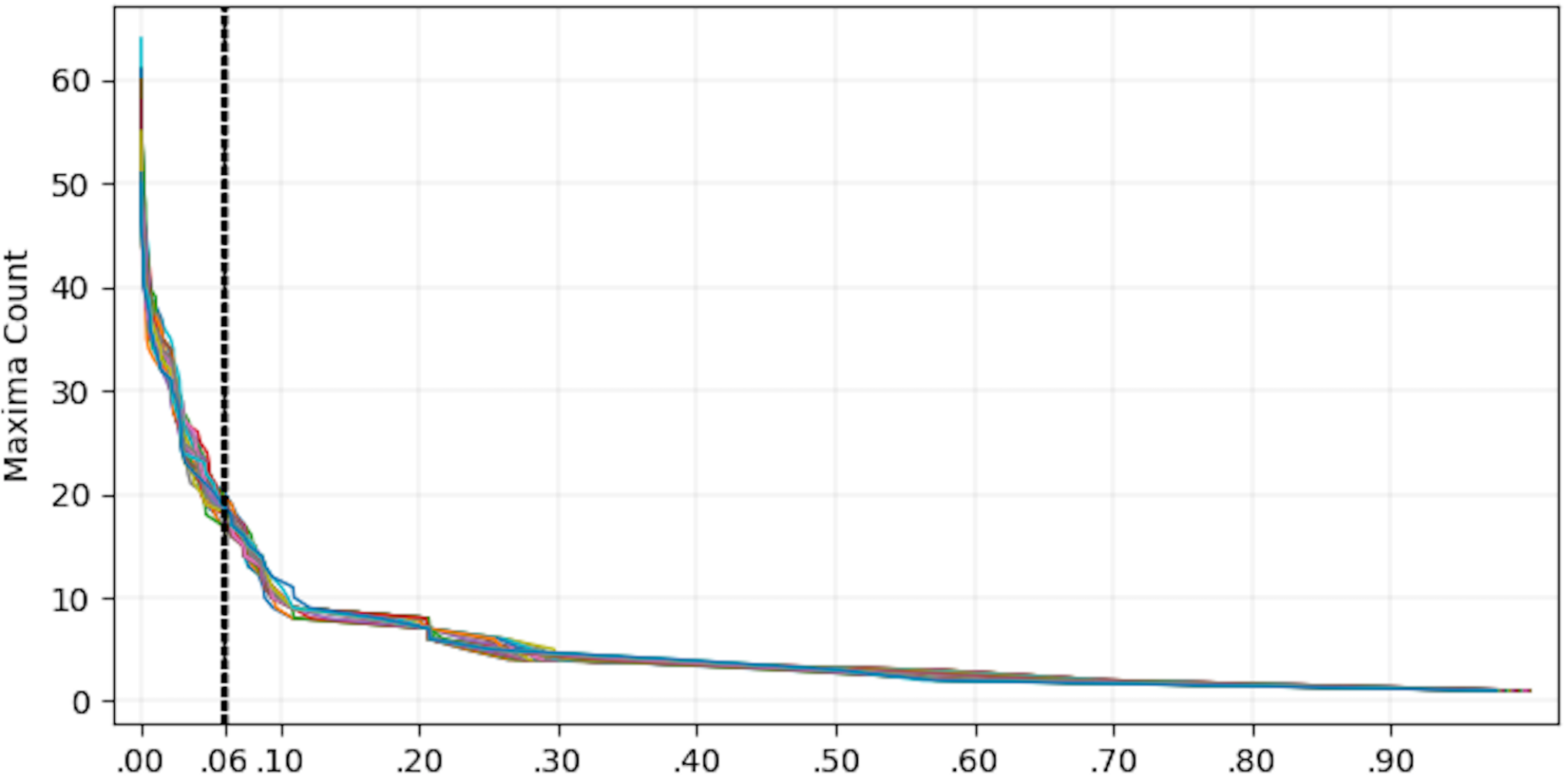} 
    \vspace{-2mm}
    \caption{{\HC}: persistent simplification $\epsilon=6\%$; x-axis is $\epsilon$.} 
    \label{fig:HC-persistence}
\end{figure}
 
\para{In-processing parameter tuning.} 
During parameter tuning, a guiding principle is to reduce oversegmentations and mismatches by minimizing $N$ and $L$.  
We introduce a parameter $L^*$ that represents an upper bound on $L$. 
In this paper, we focus on tracking features surrounding local maxima; therefore, we compute $N$ and $L$ only for local maxima trajectories. 

First, we consider parameter tuning for $W$ and $p$.  
We inspect the behavior of $W$ (or $p$) while keeping other parameters fixed. 
Through extensive experiments across all datasets in this paper, we observe that the shortest path strategy for $W$ generally behaves equal to or better than the lowest common ancestor strategy in  minimizing $N$ and $L$. 
We also observed that the uniform strategy for $p$ performs better than the parent strategy. 
Therefore, for the rest of the paper, $W$ uses the shortest path strategy and $p$ uses the uniform strategy. 

\begin{figure*}[!ht]
\centering
 \includegraphics[width=1.98\columnwidth]{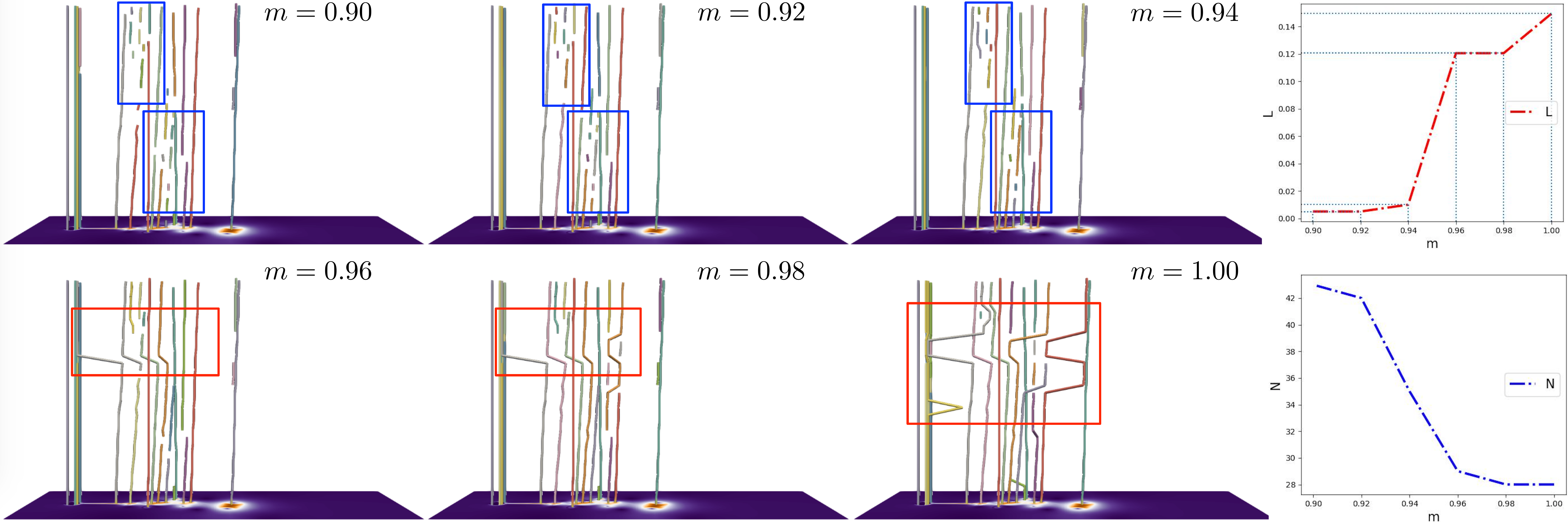}
    \vspace{-2mm}
    \caption{{\HC}. Left: for a fixed $\alpha=0.1$, perform a grid search of $m$ and observe the number of oversegmentations (in blue boxes) and mismatches (in red boxes). Right: the trend among the number of trajectories ($N$) and the maximum matched distance ($L$) as $m$ increases.} 
    \label{fig:HC-tuning}
\end{figure*} 

Second, we study the parameter tuning of $m$ for a fixed $\alpha$.  
$m$ may be considered as an in-processing step for data de-noising, by matching a certain number of features to the dummy nodes during partial optimal transport. 
We use an example in~\autoref{fig:HC-tuning} (left) to demonstrate the process. For a fixed $\alpha = 0.1$, we perform a grid search of $m \in [0.5, 1.0]$ with an increment of $0.01$.
For instance, at $m=0.90$, we see a number of oversegmented trajectories in the blue boxes; such oversegmentations decrease as $m$ increases from $0.90$ to $0.94$ (in the top row). 
On the other hand, obvious mismatches appear in the red boxes for $m \geq 0.96$ (bottom left). 
As $m$ increases from $0.9$ to $1.0$, we observe a decrease in $N$ and an increase in $L$; this is additionally demonstrated in the plots of $N$ and $L$, see~\autoref{fig:HC-tuning} (top right and bottom right).  
If our goal is to choose an appropriate \emph{global} value for $m$, then we are interested in striking a balance between  minimizing $N$ and minimizing $L$; therefore, we may choose $m=0.94$ in this example. 
However, as shown in~\autoref{fig:HC-tuning}, at $m=0.94$, there are still oversegmentations within the blue box, indicating that a \emph{locally adaptive} value of $m$ might be more appropriate in practice. 

Our final strategy aims to automatically adjust the value of $m$ between adjacent time steps to reduce $N$, without increasing $L$ drastically.   
Specifically, we perform a 2D grid search of $\alpha$ and $m$: 
$\alpha \in [0.0,1.0]$ with an increment of $0.1$, and   
$m \in [0.5, 1.0]$ with an increment of $0.01$. 
For each fixed $\alpha$, we apply the following procedure.
First, we plot the curve of $L$ as we increase $m$.
Second, we apply the elbow method and pick the elbow of the $L$ curve as an upper bound on $L$, denoted as $L^*$. 
Finally, for each pair of adjacent steps $t$ and $t+1$, we automatically choose the largest value of $m$ such that $L$ does not exceed $L^*$. 
In other words, $m$ varies adaptively across time steps, see~\autoref{fig:HC-tuning-plot} (top) with marked elbow points. 

As $\alpha$ varies, we plot the number of trajectories $N$ and the maximum matched distance $L$ ($\leq L^*$) at each $\alpha$, as shown in~\autoref{fig:HC-tuning-plot} (bottom). 
We look for a proper value of $\alpha$ to minimize both $N$ and $L$.
However, $N$ and $L$ may not be minimized at the same $\alpha$. 
In this scenario, we look for an $\alpha$ such that it minimizes $N$ while keeping $L$ to be small enough to minimize the number of mismatches. 
Using this strategy, we set $\alpha=0.1$, with a corresponding $L^*=0.00997$. 

\begin{figure}[!ht]
    \centering
    \includegraphics[width=0.98\columnwidth]{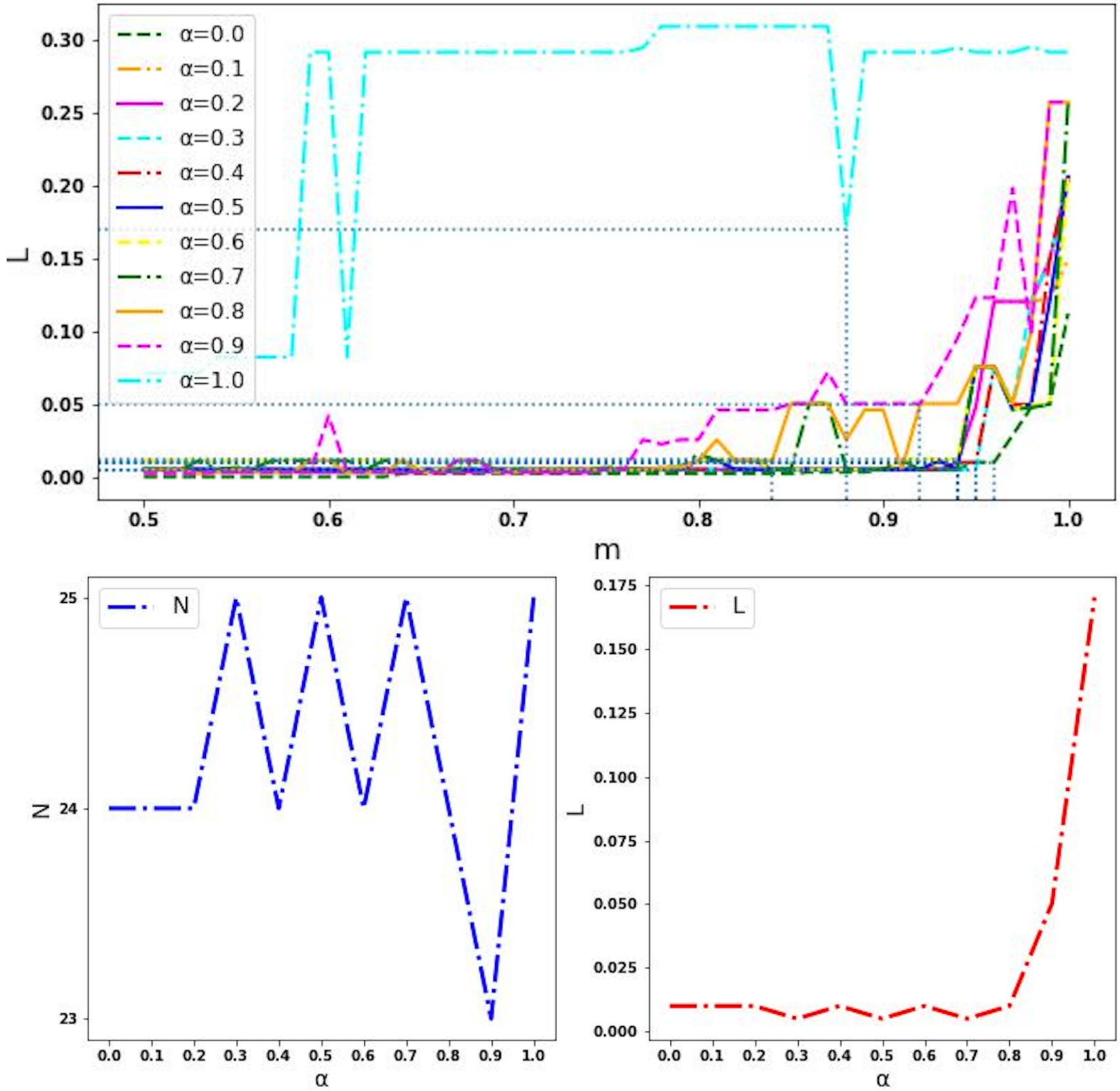} 
    \vspace{-2mm}
    \caption{{\HC}. (a) $L$ as we change the global $m$ for each $\alpha$; elbows of curves are marked with dotted horizontal and vertical lines, (b) $N$ and $L$ with respect to $\alpha$ (using adaptive $m$).} 
    \label{fig:HC-tuning-plot}
\end{figure}

\begin{figure*}[!ht]
    \centering
    \includegraphics[width=1.98\columnwidth]{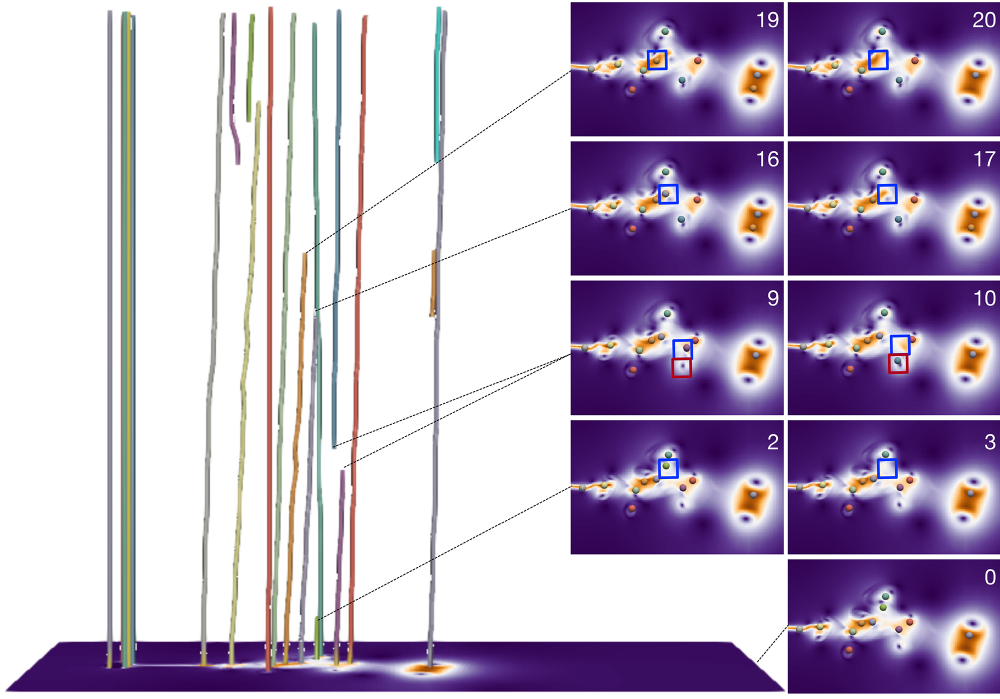}
    \vspace{-2mm}
    \caption{{\HC}. Tracking result (left) with views of scalar fields (right) that capture topological changes in the time-varying scalar field at selected time steps. The appearances and disappearances of critical points are highlighted in red and blue boxes, respectively.} 
    \label{fig:HC-tracking}
\end{figure*}

\subsubsection{Tracking Result}
\label{sec:hc-tracking-results}

\autoref{fig:HC-tracking} shows our final tracking result on the left with views of scalar fields on the right that highlight the appearances and disappearances of critical points. 
In \autoref{fig:HC-tracking} (left), the $xy$-plane visualizes the scalar field at $t=0$, and the $z$-axis shows the trajectories for all the local maxima and the global minimum as time increases. 
Most trajectories are shown to be straight lines as only minor topological changes occur in this dataset.  
Meanwhile, our framework successfully captures the appearances and disappearances of critical points. 
As shown in~\autoref{fig:HC-tracking} (right), for time steps $2 \to 3, 10 \to 11, 16 \to 17$ and $19 \to 20$, critical points disappear in the blue boxes, resulting in the termination of trajectories; for time steps $9 \to 10$, a critical point appears in a red box, resulting in the start of a new, green trajectory. 

\subsubsection{Comparison with Previous Approaches} 
\label{sec:hc-compare}

We compare the tracking results for our pFGW framework with two other state-of-the-art feature tracking approaches, referred to as  Global Feature Tracking ({GFT})~\cite{saikia2017global, saikia2017fast} and Lifted {Wasserstein} Matcher~\cite{SolerPlainchaultConche2018} ({LWM}); see the supplementary material for parameter tuning of GFT and LWM, respectively. 

\para{Implementations.}
Our pFGW framework utilizes the libraries implemented in TTK~\cite{TiernyFavelierLevine2018, gueunet2017task, lukasczyk2020localized} for merge tree computation. 
{GFT} is implemented in \textit{C++} and is available at~\cite{GFT}.  
It computes the merge trees and region segmentations, and outputs the tracking results between critical points at adjacent time steps. 
GFT allows tracking between saddles and local extrema, whereas pFGW (in our experiments) only focuses on tracking between local extrema. 
Therefore, we adjust the postprocessing of GFT to remove trajectories involving saddles. 
LWM is implemented as an embedded library in TTK.
Results of all three methods are visualized via  ParaView~\cite{AhrensGeveciLaw2005} with VTK~\cite{schroeder2006VTK}.

Since neither LWM nor GFT includes details on their parameter tuning, we apply the same parameter tuning strategy as pFGW to both LWM and GFT, that is, minimizing the number of trajectories and  the maximum matched distances; see the supplementary material for details.  

Furthermore, all three methods apply the same persistence-based simplification during preprocessing. However, since GFT is defined on a regular grid of squares, and pFGW and LWM use identical simplicial meshes, we expect minor inconsistencies on the simplified datasets between GFT and other two methods.

\para{Tracking results comparison.}
All three tracking results are shown in~\autoref{fig:HC-comparison} (top).
All three methods produce $24$ trajectories, but there are noticeable differences in GFT-produced trajectories (comparing red and blue boxes, respectively).  
We evaluate these results quantitatively based on observable oversegmentations and mismatches. 
There are obvious oversegmentations from {GFT} compared to the other two methods: a trajectory in the red box is broken in GFT, but remains continuous in pFGW and LWM. 

\begin{figure*}[!ht]
    \centering
    \includegraphics[width=1.8\columnwidth]{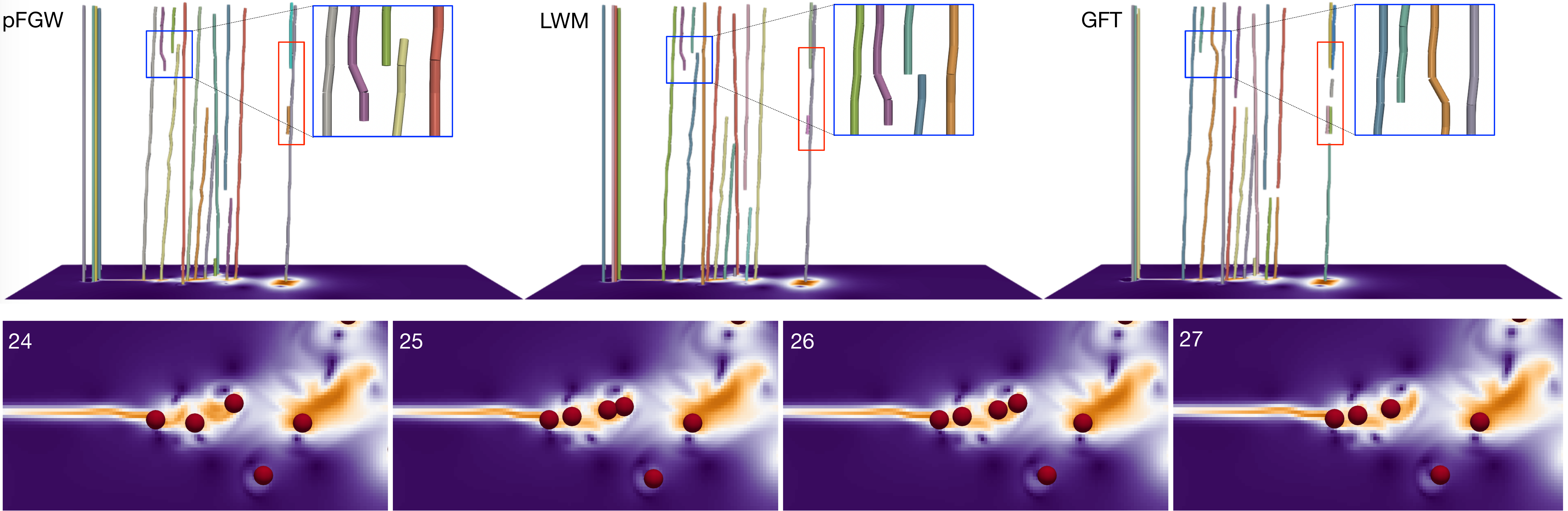}
    \vspace{-4mm}
    \caption{{\HC}. Top: from left to right, pFGW (ours), LWM, and GFT, respectively. 
    Bottom: the appearances and disappearances of local maxima in the blue boxes on top.} 
    \label{fig:HC-comparison}
    \vspace{-4mm}
\end{figure*}

As for mismatches, GFT produces a different tracking result from pFGW and LWM in the blue box, from time steps $24 \to 27$; the corresponding scalar fields are shown in~\autoref{fig:HC-comparison} (bottom).
We interpret the topological changes as follows: 
a critical point appears from $24 \to 25$, another critical point appears from $25 \to 26$, and a critical point disappears from $26 \to 27$. 
The trajectories in pFGW and LWM correctly reflect these topological changes, whereas those in GFT consider these changes to be the movements of critical points.
Therefore, pFGW and {LWM} perform similarly, but {GFT} performs slightly worse for the {\HC} dataset. 

\subsection{Unsteady Cylinder Flow}

For the {\UCF} dataset, we employ the same parameter tuning strategy detailed in~\autoref{sec:tuning}. 
We use a persistence simplification level at $\epsilon=1\%$. 
We set $\alpha=0.1$ and $L^*=0.03768$, see the supplementary material for details.  
 
\subsubsection{Tracking Results} 
 
\begin{figure}[!ht]
    \centering \includegraphics[width=1.0\columnwidth]{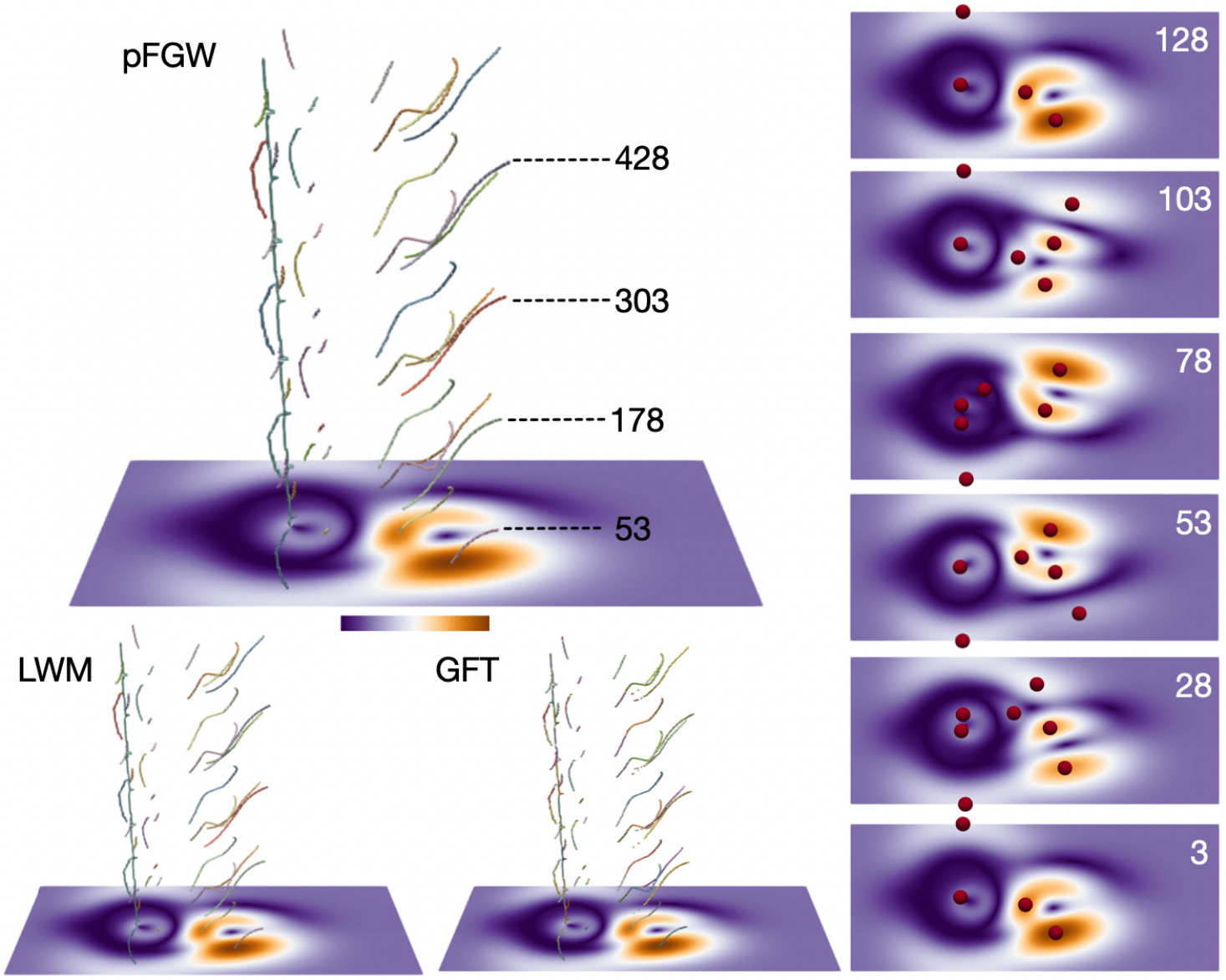}
    \vspace{-6mm}
    \caption{{\UCF}. Left: comparing tracking results for  pFGW, LWM, and GFT, respectively. 
    Right: snapshots of scalar fields within a single period between $t=3$ and $t=128$.} 
    \label{fig:UCF-comparison}
\end{figure}

Our tracking result using pFGW is highly periodic, where the extracted trajectories exhibit repetitive patterns that include the appearances, disappearances, and movements of local maxima over time, see~\autoref{fig:UCF-comparison} (left). 
We show a few time steps at $t=53, 178, 303$, and $428$ to highlight a periodicity of $\approx 125$. 
Furthermore, as shown in~\autoref{fig:UCF-comparison} (right), six snapshots show the evolution of the scalar field within a single period between $t = 3$ and $t = 128$, where the scalar field at $t = 128$ is mostly identical to the one at $t=3$. 

\subsubsection{Comparison with Previous  Approaches}
We compare our pFGW framework against the LWM and GFT methods, which give rise to $44$, $44$, and $108$ trajectories, respectively, see~\autoref{fig:UCF-comparison}.    

When considering mismatches, the trajectories from all three methods are visually similar, where there are no obvious mismatches for any of these methods. 
In particular, the (normalized) maximum matched distances across the three methods are the same, $L = 0.03768$. 

When considering oversegmentations, GFT produces $108$ trajectories, whereas pFGW and LWM each produces $44$ trajectories. 
Correspondingly, GFT shows many more broken trajectories visually in comparison with pFGW and LWM. 

\subsection{2D von K\'arman Vortex Street Dataset}
\label{subsection:vortex-street}

We then study the {\VS} dataset. 
We set $\epsilon=1\%$, $\alpha=0.1$, and $L^*=0.02537$; see the supplementary material for details.  

\subsubsection{Tracking Results}

\begin{figure}[!ht]
    \centering
    \includegraphics[width=1.0\columnwidth]{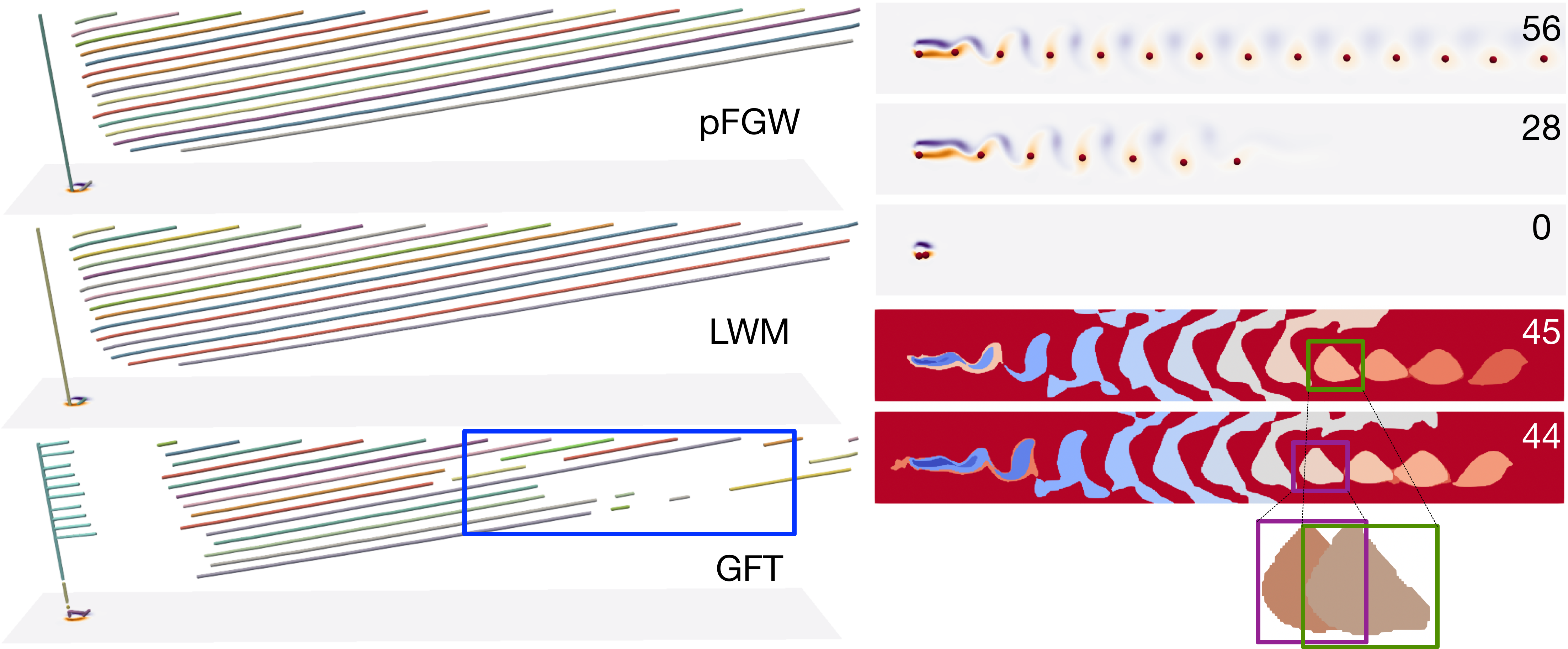}
    \vspace{-6mm}
    \caption{{\VS}. Left: comparing tracking results for  pFGW, LWM, and GFT, respectively. Right top: snapshots of scalar fields at $t=0, 28$, and $56$. Right bottom:  merge-tree-based segmentation of the scalar fields at $t=44$ and $45$. } 
    \label{fig:VS-comparison}
\end{figure}

The tracking results for {\VS} using pFGW, LWM, and GFT are shown in~\autoref{fig:VS-comparison} (left), in which there are 17, 17, and 27 trajectories, respectively. 
The results for pFGW and LWM are mostly identical, whereas the results from GFT show a number of oversegmentations  and missing trajectories at later time steps (\eg, see the blue box).  
A few snapshots of the scalar field are shown in  ~\autoref{fig:VS-comparison} (right top), where local maxima are well aligned horizontally and moving rightward at an almost constant speed. 
This characteristic leads to a large number of straight-line trajectories, as shown in~\autoref{fig:VS-comparison} (left). 
Meanwhile, a critical point remains stable in location to the left of the cylinder, whose trajectory is shown as a single straight line on the leftmost part of~\autoref{fig:VS-comparison} for both pFGW and LWM. 
 
\subsubsection{Comparison with Previous Approaches}
Our pFGW method and the LWM method perform similarly on {\VS} in terms of reducing oversegmentations and mismatches. 
Meanwhile, similar to {\HC} and {\UCF}, GFT typically introduces more oversegmentations in comparison with pFGW and LWM; in addition, certain trajectories may be missing due to insufficient feature overlaps between adjacent time steps.  
In~\autoref{fig:VS-comparison} (right bottom), we show merge-tree-based segmentation of the scalar field at time steps $44$ and $45$. 
Here, the corresponding features at $t=44$ and $t=45$ move rapidly to the right (see the purple and green boxes, respectively).  
Although the features associated with these adjacent time steps are visually similar, their overlap is quite small. 
Such insufficient feature overlaps appear to impact the tracking results significantly.   

\subsection{Ionization Front Dataset}
We study the  {\IF} dataset by setting $\epsilon = 10\%$, $\alpha=0.4$, and $L^*=0.02693$. 
A few snapshots of the scalar field at time steps $0, 30, 60$ and $90$ are shown in~\autoref{fig:IF-data}, as the instability progresses towards the right. 
\begin{figure}[!ht]
    \centering
    \includegraphics[width=1.0\columnwidth]{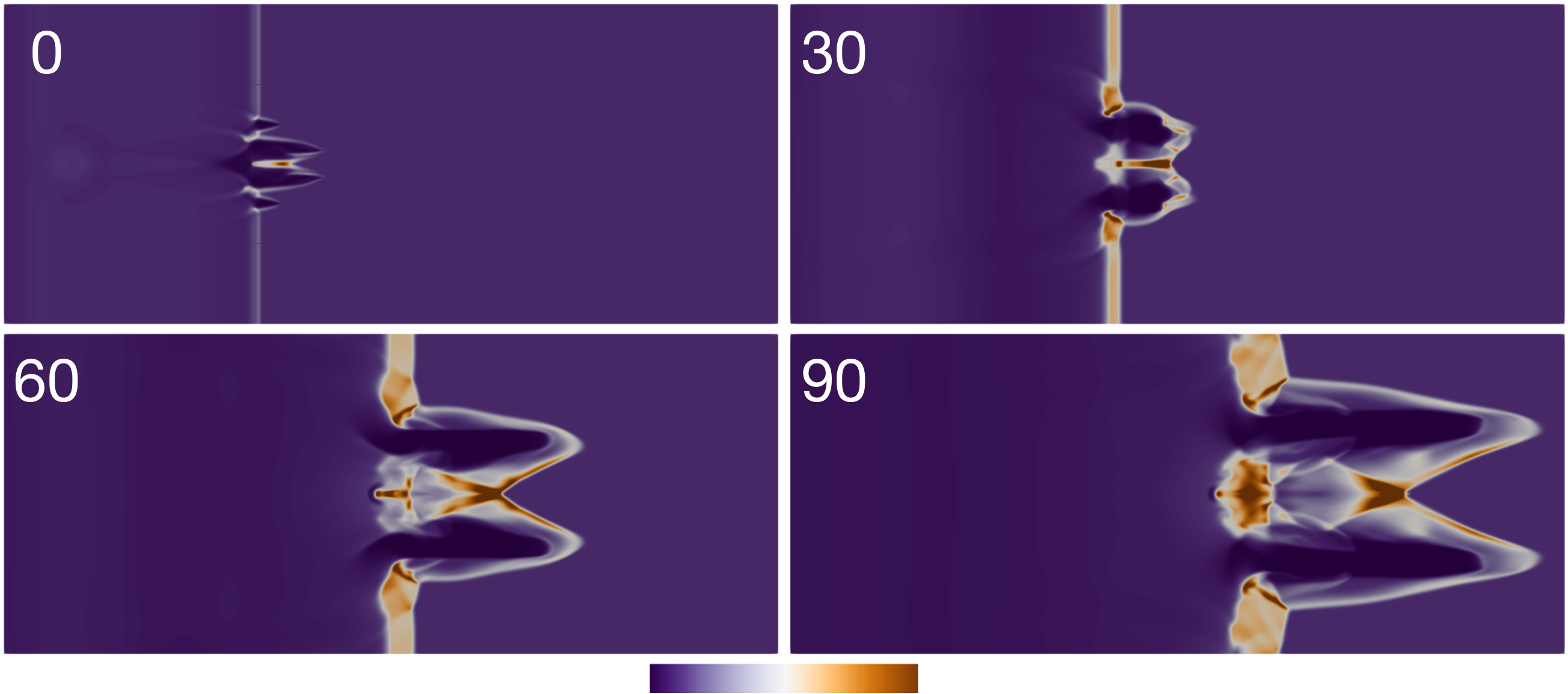}
    \vspace{-8mm}
    \caption{A few snapshots of {\IF} dataset.} 
    \label{fig:IF-data}
\end{figure}

\subsubsection{Tracking Results}

\begin{figure}[!ht]
    \centering
    \includegraphics[width=0.98\columnwidth]{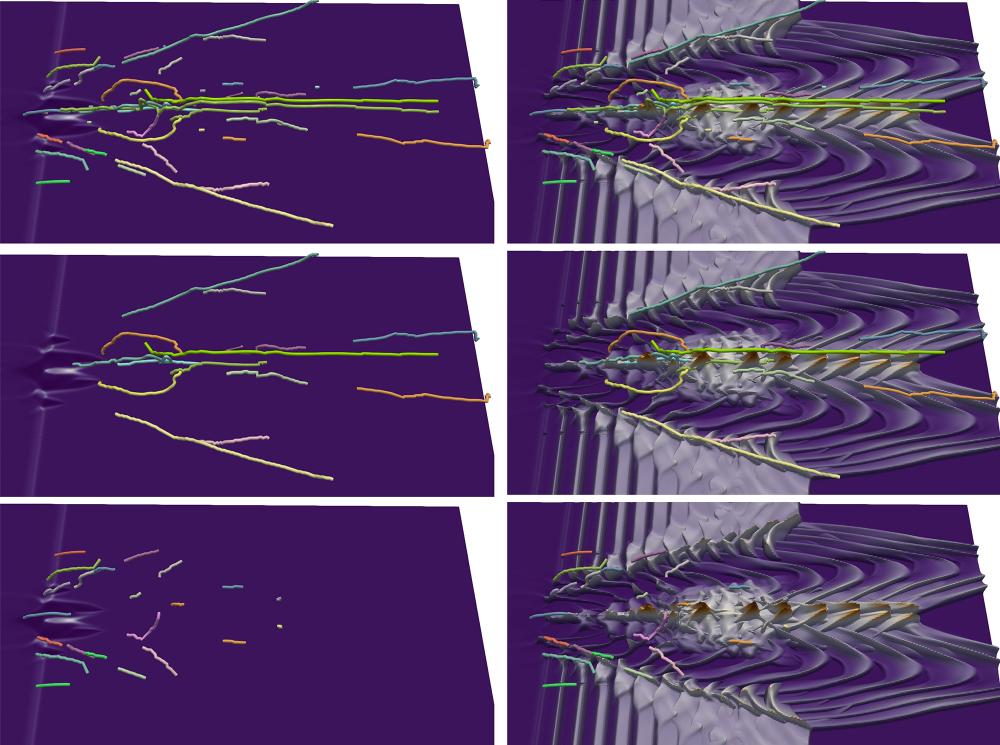}
    \vspace{-2mm}
    \caption{{\IF}. Left: pFGW trajectories are shown with the scalar field at time step $0$. Right: pFGW trajectories are visualized with the landscape of the time-varying scalar field. Top: all trajectories; middle: long-term trajectories; bottom: short-term trajectories. } 
    \label{fig:IF-tacking}
\end{figure}

We demonstrate our pFGW tracking results in~\autoref{fig:IF-tacking} (left), where trajectories are shown with the scalar field at $t=0$. 
We then visualize these trajectories with the landscape of the time-varying scalar field in~\autoref{fig:IF-tacking} (right), which is constructed by stacking the original scalar field at  time steps $0$, $10$, $20$, …, $100$, and $110$. 
Such a landscape clearly shows the rightward propagation of the ionization front. 
The results shown in~\autoref{fig:IF-tacking} (top) thus contain a number of trajectories that capture such a trend.  

We further split these trajectories into two sets: trajectories that last longer than $29$ time steps (\emph{long-term trajectories}) in~\autoref{fig:IF-tacking} (middle), and those that last between $5$ and $29$ steps (\emph{short-term trajectories}) in~\autoref{fig:IF-tacking} (bottom). 
We ignore trajectories shorter than $5$ time steps as they do not capture the global trend of the data.  
A number of the long-term trajectories appear to follow the direction of the radiation waves, whereas some short-term trajectories capture local interactions among them.  

\subsubsection{Comparison with Previous Approaches}
In terms of oversegmentations, pFGW, LWM, and GFT give rise to 51, 52, and 92 trajectories, respectively.  
pFGW produces slightly fewer trajectories than LWM, whereas GFT oversegments and produces the largest number of trajectories, see~\autoref{fig:IF-comparison}. 
In particular, GFT produces noticeably broken long-term trajectories, implying that it fails to track some major features consistently. 

In terms of mismatches, trajectories from all three methods interpret the evolution of features in a similar way. However, pFGW produces the smallest (normalized) maximum distance of 0.02693, whereas LWM and GFT give rise to a (normalized) maximum distance of 0.03840. 

\begin{figure}[!ht]
    \centering
    \includegraphics[width=1.0\columnwidth]{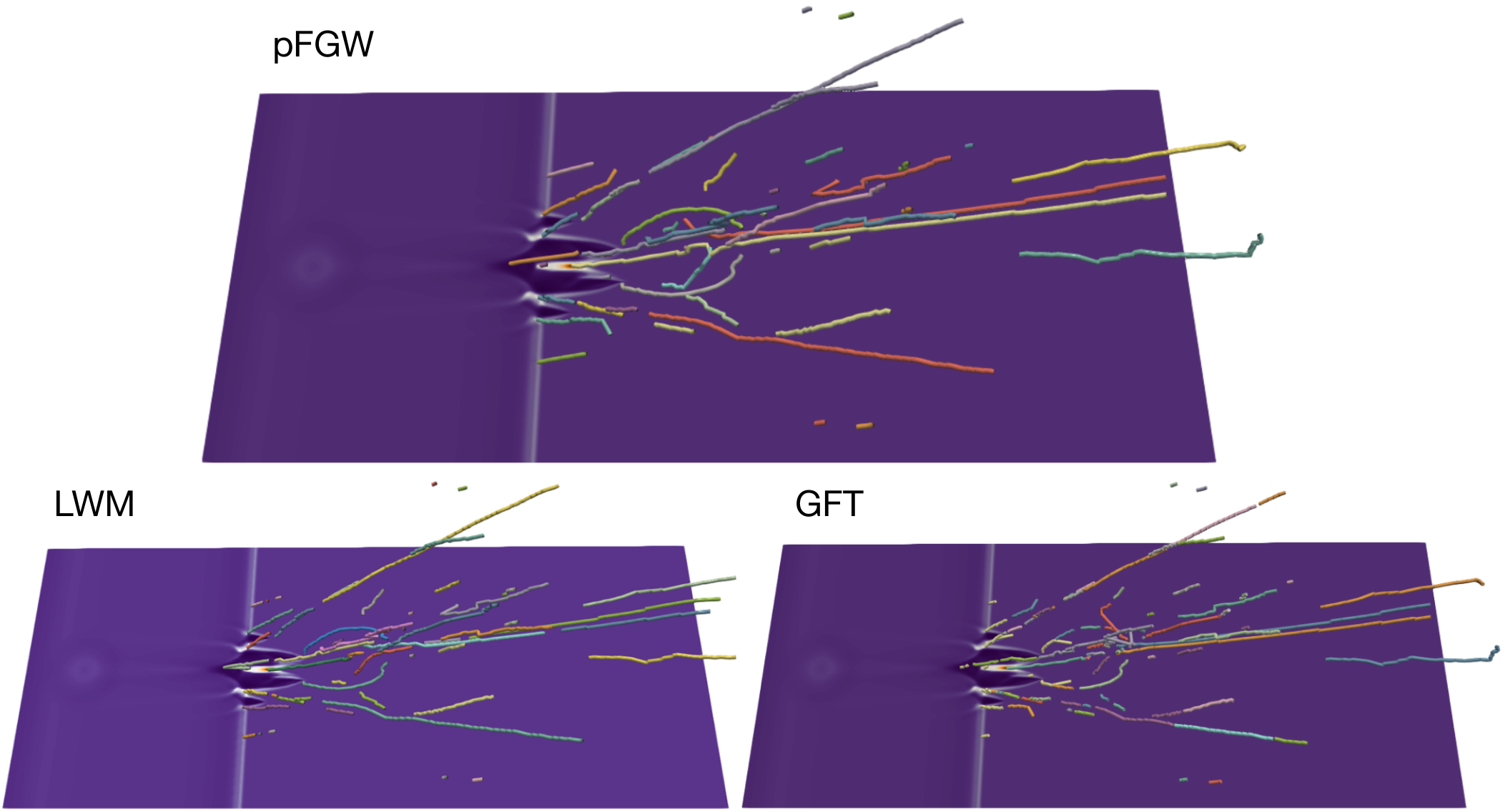}
    \vspace{-6mm}
    \caption{{\IF}. Comparing tracking results for pFGW, LWM, and GFT, respectively.} 
    \label{fig:IF-comparison}
    \vspace{-6mm}
\end{figure}

\subsection{Cloud Dataset}
\label{sec:cloud}

\subsubsection{Tracking Results }
For our cloud tracking task, \autoref{fig:CL-thickness} (left) illustrates the scalar field of interest, namely, the cloud optical thickness (at the 1st time step), where areas with high cloud optical thickness are shown in white, orange and brown.   
We track the cloud by tracking the movement of the local maxima of such a field. 

\begin{figure}[!ht]
    \centering
    \includegraphics[width=1.0\columnwidth]{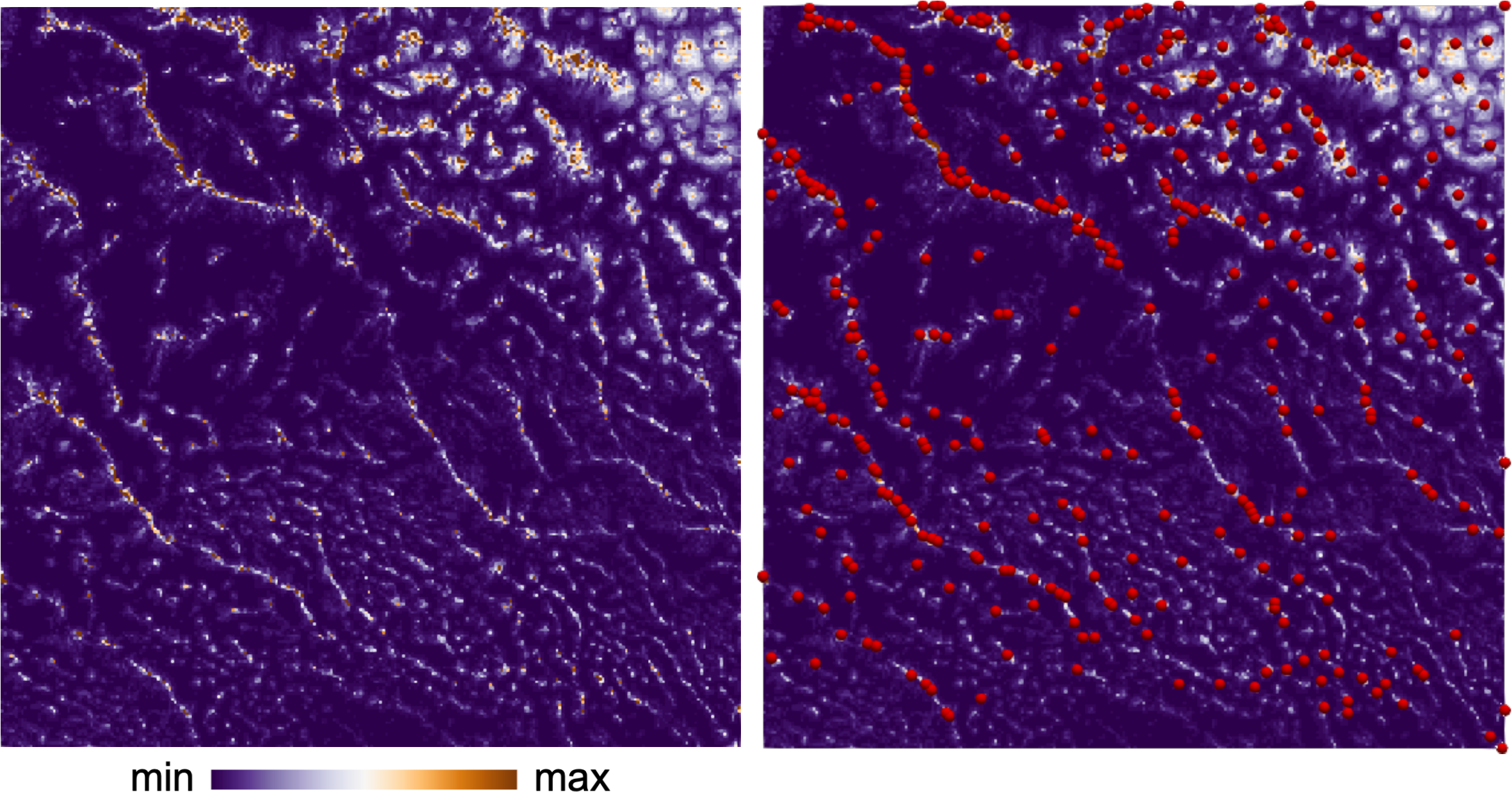}
    \vspace{-6mm}
    \caption{{\CL}. Visualization of a cloud optical thickness field (left) along with its local maxima in red (right).} 
    \label{fig:CL-thickness}
\end{figure}

As shown in \autoref{fig:CL-comparison} (right), local maxima (in red) are densely distributed in areas with high cloud optical thickness. 
Such characteristics present multiple challenges for feature tracking.  
First, the regions that contain local maxima may be very small, leading to insufficient feature overlaps between adjacent time steps. 
Second, densely distributed features frequently appear and disappear, making it challenging to track individual features. 

We present the tracking results in~\autoref{fig:CL-comparison}, in which the majority of the features move toward the left side of the (observable) domain.
For pFGW, we set $\epsilon=50\%$, $\alpha=0.1$ and $L^*=0.0653$. 

\begin{figure}[!ht]
    \centering
    \includegraphics[width=1.0\columnwidth]{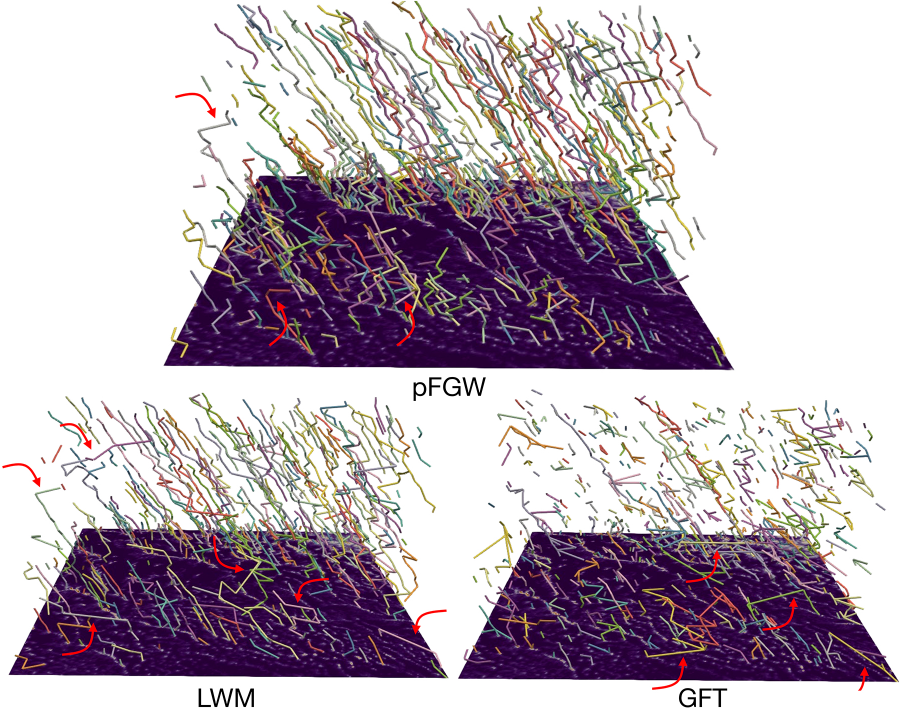}
    \vspace{-6mm}
    \caption{{\CL}. Tracking results for GFT, LWM, and pFGW, respectively. Red arrows highlight some (not all) trajectories containing mismatches.} 
    \label{fig:CL-comparison}
    \vspace{-2mm}
\end{figure}

\subsubsection{Comparison with Previous Approaches}
Overall, pFGW produces the largest number of  trajectories and the smallest number of isolated local maxima (i.e., trajectories that last for a single time step), comparing with the other two approaches. 
It produces $584$ trajectories plus $274$ isolated local maxima. 
In contrast, GFT produces $476$ trajectories plus $1390$ isolated local maxima, whereas LWM produces $302$ trajectories plus $1236$ isolated local maxima. 

All three methods produce mismatches, as indicated by red arrows in~\autoref{fig:CL-comparison}.  
However, pFGW produces long trajectories with the fewest number of mismatches. 
In particular, pFGW has the best maximum matched distance in comparison with LWM and GFT, respectively. 
Statistically, the largest (normalized) maximum matched distances in the results of pFGW, LWM, and GFT are  $0.0653$, $0.1828$, and $0.2473$, respectively. 
\autoref{fig:CL-matched-distance} displays the distributions of maximum matched distances from trajectories across three methods, where pFGW produces zero trajectories with a maximum matched distance larger than $0.075$. 
This shows that pFGW is comparatively most resistant to mismatches. 

\begin{figure}[!ht]
    \centering
    \includegraphics[width=1.0\columnwidth]{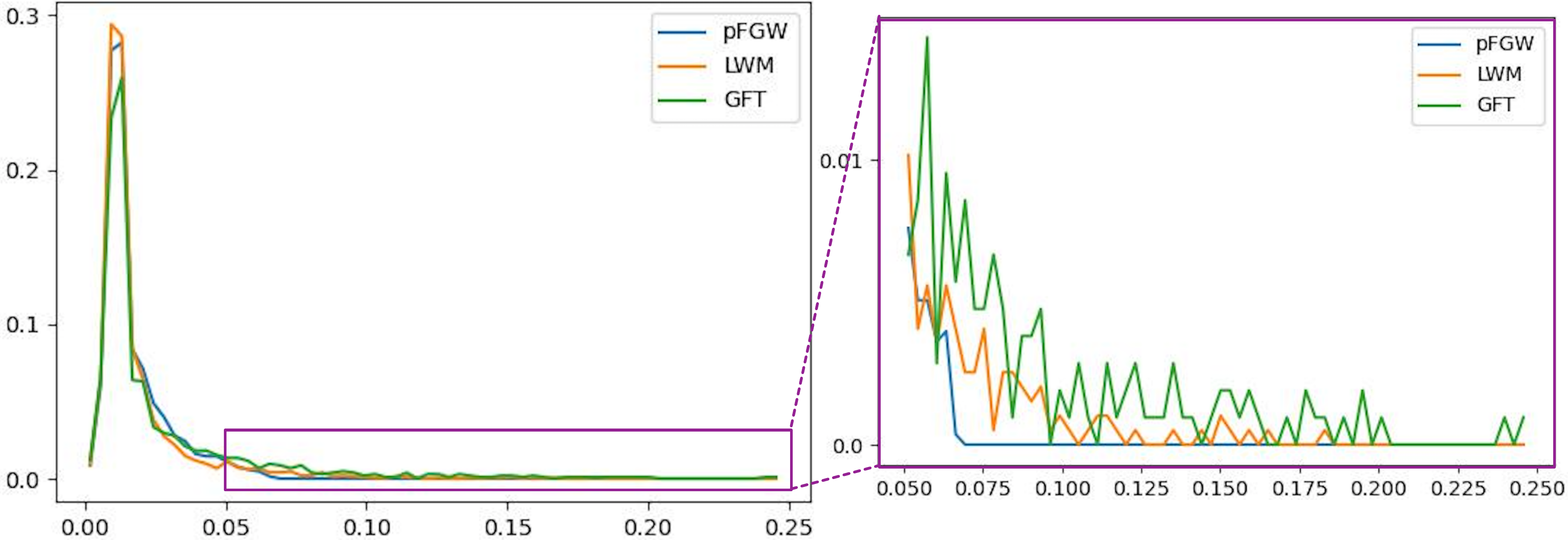}
    \vspace{-6mm}
    \caption{{\CL}. Left: distributions of maximum matched distances of trajectories obtained with pFGW, LWM, and GFT, respectively; x-axis is the maximum matched distance, y-axis is the percentage of the number of trajectories. Right: a zoomed-in view of the tail of the distribution.} 
    \label{fig:CL-matched-distance}
\end{figure}

Upon further inspection, GFT suffers severely from oversegmentations, that is, it produces many short trajectories and isolated local maxima. 
This is likely due to insufficient feature overlaps between adjacent time steps. 
LWM also fails to find trajectories for many local maxima, resulting in a large number of isolated local maxima. 
Under the current configuration of LWM, the cost of matching a local maximum to its diagonal projection (causing the disappearance of a feature) is dominated by the distance between the maximum and its pairing saddle. 
When a local maximum is very close to its pairing saddle, LWM tends to match the local maximum to its diagonal projection (rather than looking for its corresponding local maximum in an adjacent time step), which leads to an unpaired local maximum. 
On the other hand, in LWM, as a large number of features appear and disappear, mismatches between local maxima occur when the cost of matching a local maximum to its diagonal projection (causing the disappearance of a feature) outweighs the cost of matching the same local maximum to a faraway local maximum (causing mismatches). 
Whereas LWM relies on the locations of pairing  saddles to determine the appearances and disappearances of local maxima, our pFGW approach imposes more constraints on the relations among the critical points using merge trees, thus producing the smallest number of mismatches comparatively.

\subsection{3D Isabel Dataset}

\begin{figure}[!ht]
    \centering
 \includegraphics[width=0.98\columnwidth]{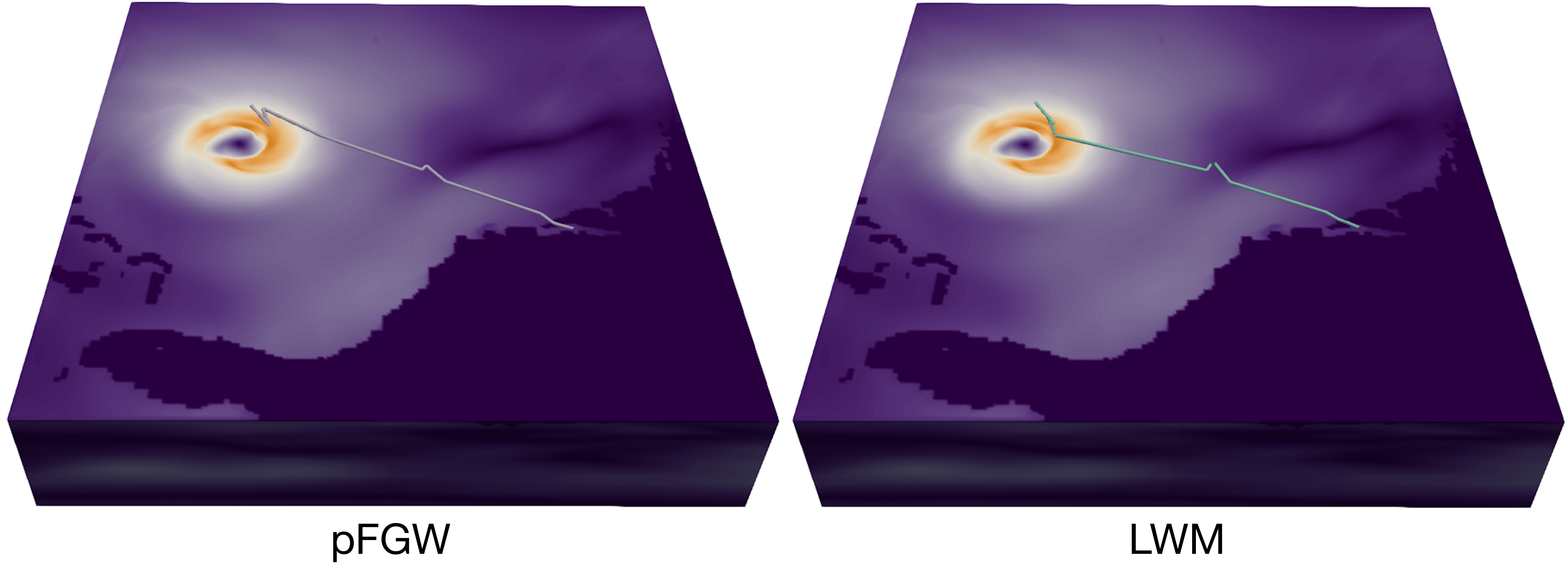}
    \vspace{-2mm}
    \caption{{\IS}. Tracking results for pFGW (left) and LWM (right).} 
    \label{fig:IS-tracking}
\end{figure}

For the 3D {\IS} dataset, we apply both pFGW and LWM to track the trajectory of the global maximum, which highlights the movement of the main hurricane. 
This dataset contains a discrete set of time steps with large gaps; thus, it is not suitable for feature tracking based on region overlaps (such as GFT). 
For pFGW, we use $\epsilon=10\%$, $\alpha=0.6$, $L^* = 0.4010$.  
As shown in~\autoref{fig:IS-tracking}, both pFGW and LWM successfully track the movement of the hurricane. 
These results highlight the robustness of topology-based feature tracking in 3D.

\subsection{3D Viscous Fingering Dataset}
\label{sec:viscous-finger}

\subsubsection{Tracking Results}

We focus on trajectories below the water surface for the 3D {\VF} dataset. 
The tracking results are shown in~\autoref{fig:VF-comparison} (top).  
For pFGW, we set $\epsilon = 1\%$, $\alpha = 0.1$, and $L^* = 0.1369$.
Due to high feature density, we highlight long-term trajectories (that last at least 20 time steps) in~\autoref{fig:VF-comparison} (bottom).   

\begin{figure}[!ht]
    \centering
    \includegraphics[width=1.0\columnwidth]{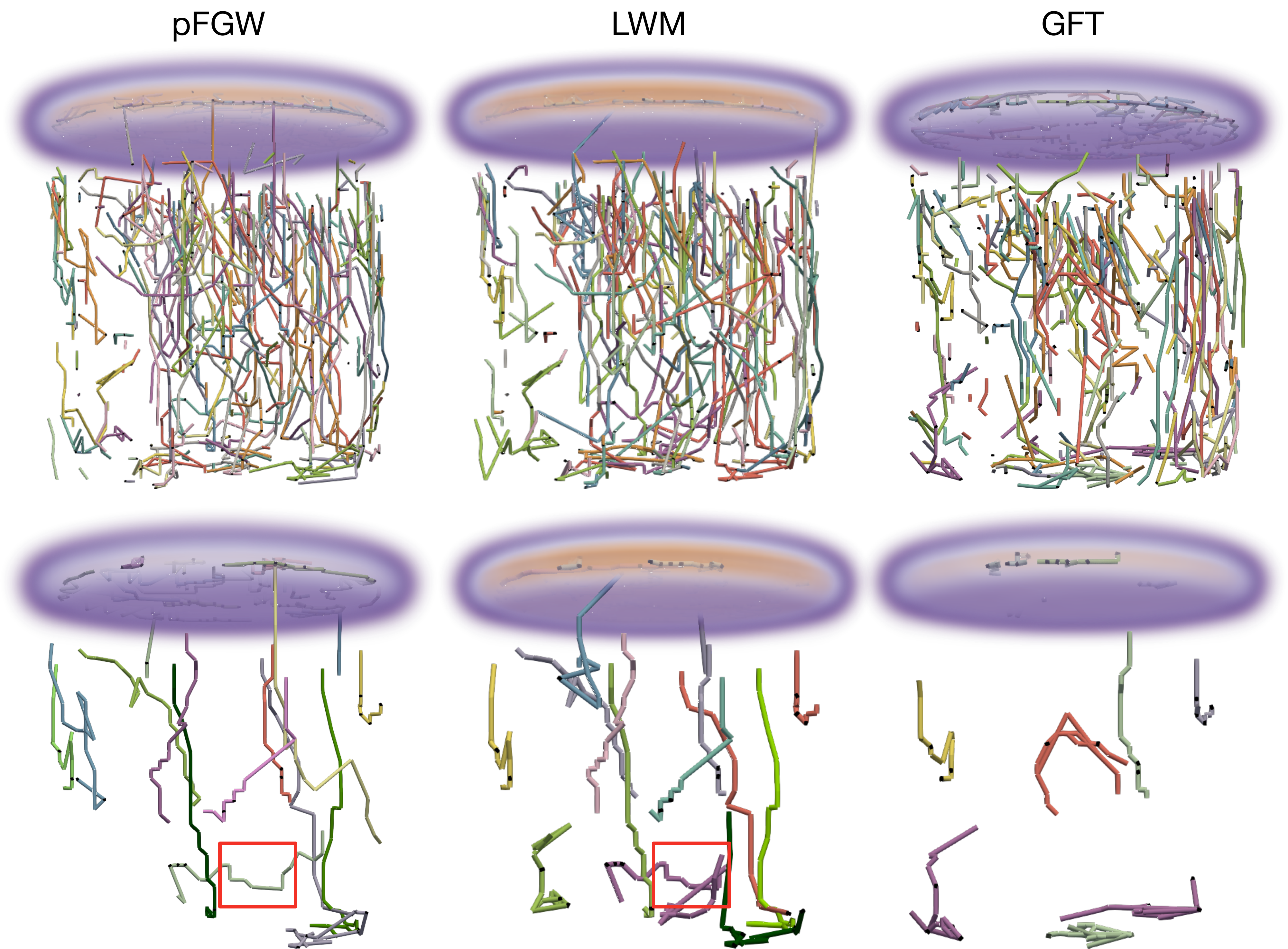}
    \vspace{-6mm}
    \caption{{\VF}. Top: tracking results for GFT, LWM, and pFGW, respectively. Bottom: trajectories that last at least 20 time steps. Red boxes contain a region of interest for analyzing mismatches.}
    \label{fig:VF-comparison}
    \vspace{-4mm}
\end{figure}

\subsubsection{Comparison with Previous Approaches}

As shown in~\autoref{fig:VF-comparison} (bottom), GFT does not produce as many long-term  trajectories as in the case of pFGW and LWM. 
pFGW produces the smallest number of isolated local maxima. 
In total, pFGW gives 424 trajectories plus 158 isolated local maxima; LWM produces 283 trajectories plus 619 isolated local maxima; and GFT leads to 462 trajectories plus 378 isolated local maxima. 
Even though LWM produces fewer trajectories than pFGW, these trajectories contain more mismatches that incorrectly connect faraway local maxima. 
Statistically, pFGW produces the best maximum matched distance: the largest (normalized) maximum distances for pFGW, LWM, and GFT are $0.1369$, $0.3214$, and $0.3499$, respectively. 

We give a case study where LWM produces mismatches that incorrectly connect faraway local maxima in~\autoref{fig:VF-case-study}. 
Here, we compare trajectories in a region of interest  enclosed by red boxes in~\autoref{fig:VF-comparison}.  
Across time steps $71 \to 76$, local maxima tracked by the same trajectories are colored the same. 
In~\autoref{fig:VF-case-study} (left), pFGW correctly identifies three trajectories (purple, yellow, and green), including the long-term green trajectory.  
In~\autoref{fig:VF-case-study} (right), LWM incorrectly tracks these local maxima along the purple trajectory. 

For example, in LWM, the purple trajectory between time steps $71 \to 72$ and $75 \to 76$ contain mismatches because the matched local maxima belong to different superlevel set components with minimum overlaps. 
In addition, LWM also terminates the green trajectory too early because the white maximum at time step $75$ and the purple maximum at time step $76$ belong to the same superlevel set component. 
In comparison, pFGW produces more reasonable trajectories: local maxima of different superlevel set components are not connected to the same trajectory; and the continuity of the green trajectory is preserved. 

\begin{figure}[!ht]
    \centering
    \vspace{-4mm}
    \includegraphics[width=0.7\columnwidth]{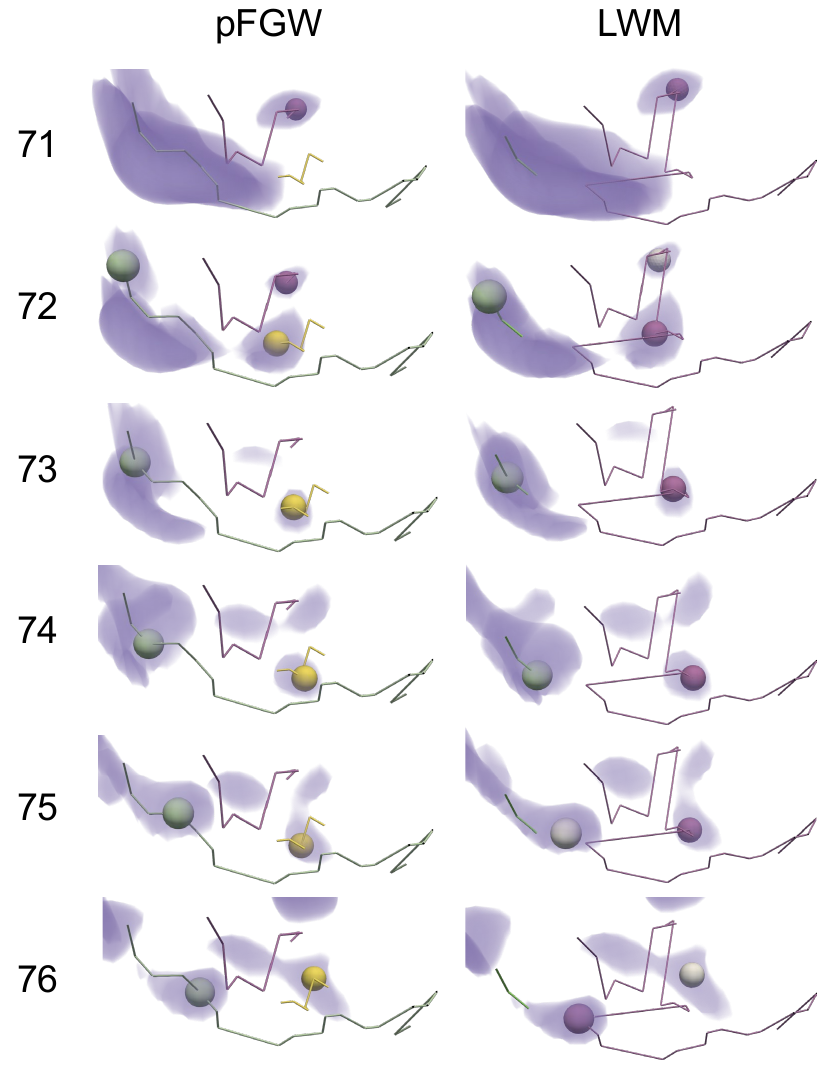}
    \vspace{-4mm}
    \caption{{\VF}. Comparing trajectories extracted by pFGW and LWM in a region of interest enclosed by red boxes in~\autoref{fig:VF-comparison}.}
    \label{fig:VF-case-study}
\end{figure}
 
In this case study for LWM, the distances between the mismatched local maxima are quite small. 
Furthermore, these critical points all have small persistence. 
Therefore, persistence information and critical point locations are not sufficient for LWM to avoid these mismatches, where as pFGW performs better by imposing additional topological constraints on the critical points via merge trees.

\subsection{Runtime Analysis}
We perform runtime analysis for all three approaches  (GFT, LWM, and pFGW) under the fine-tuned parameter configurations, as shown~\autoref{table:runtime}. 
All results are obtained from a laptop with a 12th Gen Intel(R) Core(TM) i9-12900H 2.50 GHz CPU with 32 GB memory. Both GFT and LWM are implemented in C++, whereas the pFGW is implemented in Python.
For the {\UCF}, {\VS}, {\IF}, {\IS}, and {\VF} datasets, pFGW achieves a similar runtime with LWM. 
GFT is generally slower than the other two methods except on the {\VF} dataset.  
pFGW is the slowest among the three for the {\HC} dataset. 
Overall, all three methods take $\leq 0.01$ second to compute the feature matching between a pair of adjacent time steps when the number of nodes is below $100$.
We do not include the runtime for merge tree generation as it is part of the data preprocessing. 
GFT spends more time on merge tree generation than LWM and pFGW, since it requires extra information on merge tree segmentation. 

In terms of computational complexity, minimizing the pFGW distance between merge trees requires $\mathcal{O}(n_1 n_2^2 + n_1^2 n_2)$ per iteration, where $n_1$ and $n_2$ are the size of merge trees. In our experiments, the pFGW distance converges within $20$ iterations for all datasets.

\begin{table}[!ht]
\resizebox{1.0\columnwidth}{!}{
\begin{tabular}{c|c|c|c|c|c}
\hline
\textbf{Dataset} & \# of nodes & Time steps &  Method & Total time (sec) & Avg. time \\ \hline
\multirow{3}{*}{\HC}  & \multirow{3}{*}{40} & \multirow{3}{*}{31}  & GFT & 0.120  & 0.0040 \\
 & &  & LWM & 0.045  & 0.0015 \\
 & &  & pFGW & 0.146  & 0.0049 \\ \hline
\multirow{3}{*}{\UCF} & \multirow{3}{*}{16} & \multirow{3}{*}{499} & GFT & 1.049  & 0.0021 \\
 & &  & LWM & 0.325  & 0.0007 \\
 & &  & pFGW & 0.414  & 0.0008 \\ \hline
\multirow{3}{*}{\VS}  & \multirow{3}{*}{30} & \multirow{3}{*}{59}  & GFT & 0.148  & 0.0026 \\
 & &  & LWM & 0.095  & 0.0016 \\
 & &  & pFGW & 0.098  & 0.0017 \\ \hline
\multirow{3}{*}{\IF} & \multirow{3}{*}{40} & \multirow{3}{*}{123} & GFT & 0.577  & 0.0047 \\
 & &  & LWM & 0.333  & 0.0027 \\
 & &  & pFGW & 0.318  & 0.0026 \\ \hline
\multirow{3}{*}{\CL} & \multirow{3}{*}{769} & \multirow{3}{*}{34} & GFT & 4.180  & 0.1267 \\
 & &  & LWM & 1.464  & 0.0444 \\
 & &  & pFGW & 4.043  & 0.1225 \\ \hline
\multirow{3}{*}{\IS} & \multirow{3}{*}{14} & \multirow{3}{*}{12}  & GFT & 0.378 & 0.0344  \\
 & &  & LWM & 0.176  & 0.0160 \\
 & &  & pFGW & 0.012  & 0.0011 \\ \hline
\multirow{3}{*}{\VF} & \multirow{3}{*}{96} & \multirow{3}{*}{120} & GFT & 0.879  & 0.0074 \\
 & &  & LWM & 1.041  & 0.0087 \\
 & &  & pFGW & 1.036  & 0.0087 \\ \hline
\end{tabular}}
\vspace{2mm}
\caption{Runtime (in seconds) of feature tracking with {GFT}, {LWM}, and pFGW, respectively. Average time is computed for a pair of adjacent time steps.}
\label{table:runtime}
\vspace{-8mm}
\end{table}

\begin{figure*}[ht]
    \centering  
    \includegraphics[width=1.8\columnwidth]{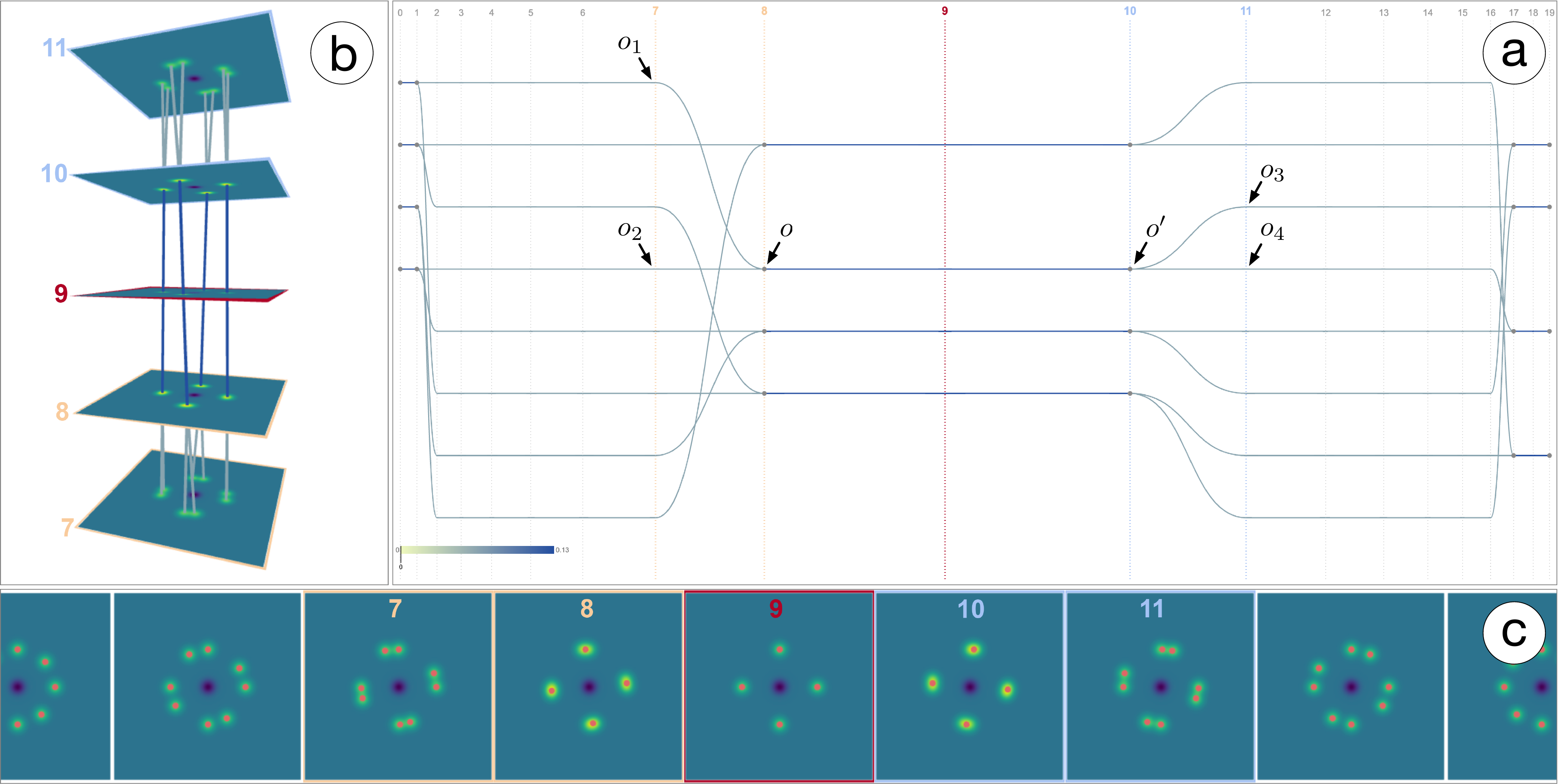}
    \vspace{-2mm}
    \caption{A visual demo of probabilistic feature tracking: 
    the graph view (a), the track view (b), and the data view (c).} 
    \label{fig:tracking-interface}
    \vspace{-2mm}
\end{figure*}

\begin{figure*}[ht]
    \centering
    \includegraphics[width=1.8\columnwidth]{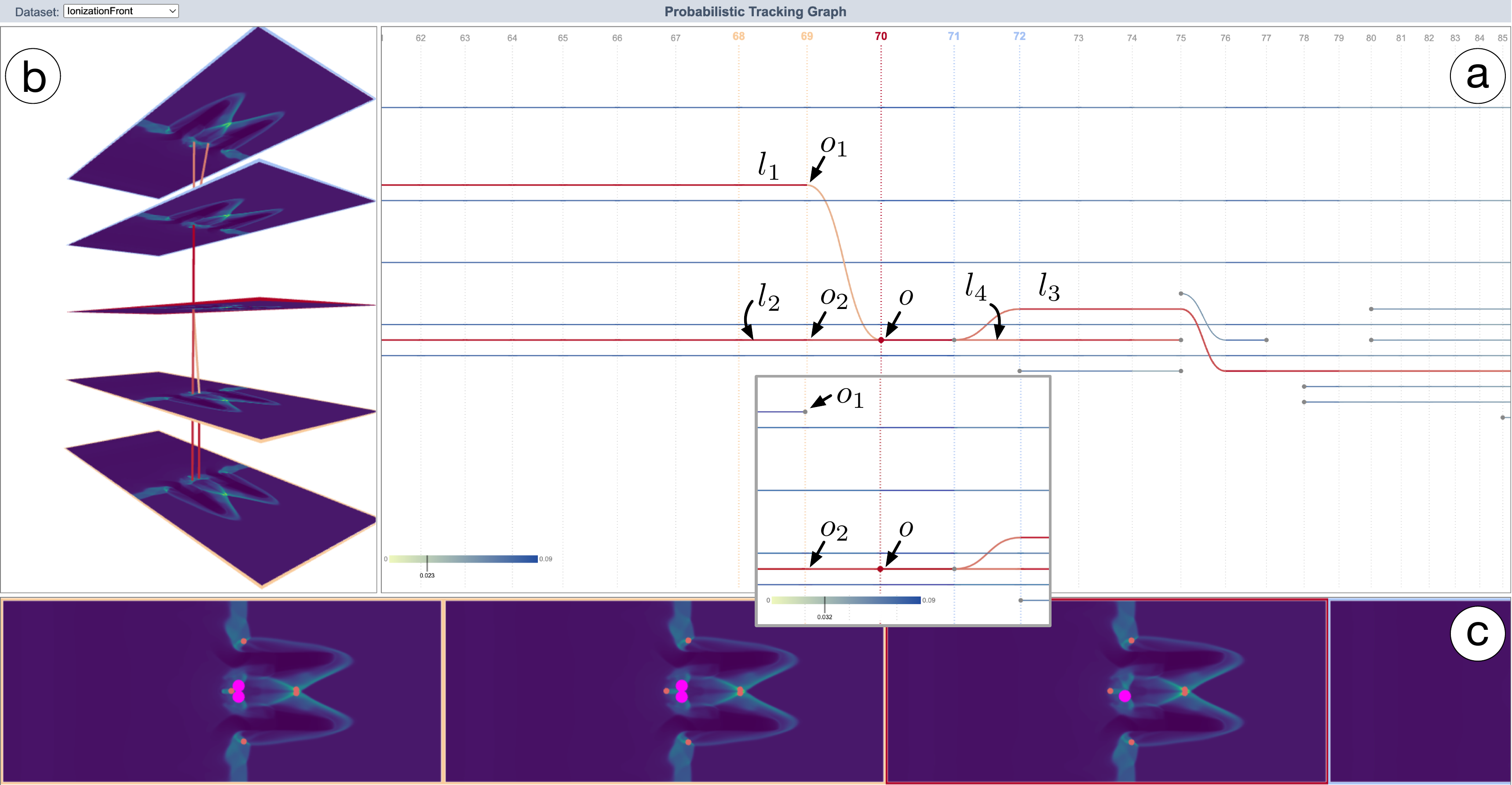}
    \vspace{-2mm}
    \caption{Selecting a particular feature $o$ to highlight relevant tracks (in red) and features (in magenta). The visual demo showcases uncertainty in tracking across different probability thresholds (see the box insert).} 
    \label{fig:tracking-focus-feature}
    \vspace{-2mm}
\end{figure*}

\section{Probabilistic Tracking Graphs}
\label{sec:tracking-graphs}

A direct consequence of our pFGW method is that it enables richer representations of tracking graphs, referred to as \emph{probabilistic tracking graphs}. 
The partial optimal transport provides a probabilistic coupling between features at adjacent time steps, which are then visualized by weighted tracks of these tracking graphs. 

We provide a visual demo for probabilistic feature tracking for several 2D time-varying datasets. 
We illustrate its visual interface using a synthetic dataset. 
As shown in~\autoref{fig:tracking-interface}, the synthetic  dataset is constructed as a mixture of nine Gaussian functions: one negative Gaussian function stays fixed at the center, eight Gaussian functions are positioned on a cycle, four of which remain stationary, whereas the other four move clockwise around the center. 
We focus on tracking the local maxima across time. 

As shown in~\autoref{fig:tracking-interface}, the \emph{graph view} {(a)} visualizes a tracking graph whose feature tracks are equipped with probabilistic tracking information. 
The \emph{track view} {(b)} displays tracks across five consecutive time steps in 3D spacetime centered around  the selected time step. The \emph{data view} {(c)} presents the scalar fields at the same five time steps.  
With multiple views, users can explore the probabilistic feature tracking results from global and local  perspectives.

The graph view {(a)} visualizes a tracking graph that captures the evolution of features across all time steps.
Vertical \emph{time bars} are positioned along the x-axis to represent time steps in increasing order, whereas \emph{tracks} associated with individual features are laid out horizontally in a way that minimizes edge crossings. 
Nodes at the intersection of time bars and tracks represent features that appear, disappear, merge, or split. 
If feature $i$ from time $t$ is coupled with feature $j$ from time $t+1$ with a nonzero measure in the coupling matrix $C$, an edge is drawn between these two features in the tracking graph, whose color and opacity encodes the value of $C(i,j)$ as indicated in the color bar. 
$C(i,j)$ is a probability measure, where higher value implies a higher probability of matching feature $i$ with feature $j$. 
Users could filter the tracks in the tracking graphs by scrolling the color bar, in order to explore the tracking graphs at different probability thresholds.

In the tracking graph shown in~\autoref{fig:tracking-interface} (a), when the four moving Gaussian functions coincide with the four stationary ones, their corresponding features merge together at time step $8$.
Subsequently, these merged features split at time step $10$. 
The tracking graph depicts such events as probabilistic merges and splits. 
From time steps $7 \to 8$, two features $o_1$ and $o_2$ merged into feature $o$ with the equal probability. 
Similarly, from time steps $10 \to 11$, a single feature $o'$ (that corresponds to $o$) splits into two ($o_3$ and $o_4$) with equal probability. 
These merging and splitting events are also encoded in the data view and the track view, see~\autoref{fig:tracking-interface} (b) and (c). 
In this case, $o_2$ at time step $7$ corresponds to $o_4$ at time step $11$, $o$ at time step $8$ corresponds to $o'$ at time step $10$; however, $o_1$ at time step $7$ does not match to $o_3$ at time step 11, since there are ambiguities in matching due to symmetry.

We now visually demonstrate the probabilistic tracking graph for the {\IF} dataset. In the example shown in~\autoref{fig:tracking-focus-feature}, we set the probability threshold at $0.023$ (main view) and $0.032$ (inserted view), respectively. 
To investigate the data of interest, users could select a specific time step $t$ by clicking its corresponding time bar, which updates views {(a)}, {(b)}, and {(c)}. 
For the graph view {(a)}, the selected time bar is highlighted in red, with the previous two (at $t-2$, $t-1$) and subsequent two (at $t+1$, $t+2$) time bars colored in orange and blue, respectively. 
We smoothly adjust intervals between time steps based on the fisheye technique using animation, so that the focus area surrounding the selected time bar is magnified and the area away from the focus is compressed. 
For the track view {(b)}, we render five scalar fields centered by the selected time step and highlight the tracks  among them in a 3D spacetime, while supporting zooming and rotation. 
For the data view {(c)}, we visualize the five 2D scalar fields side by side, where tracked features are highlighted in red. 

Furthermore, our visual demo allows users to explore tracks associated with specific features.  As illustrated in ~\autoref{fig:tracking-focus-feature}, users can select a feature of interest (denoted by $o$), which sits at the intersection of four tracks $l_1, l_2, l_3$, and $l_4$; all of which are highlighted in red while maintaining their   opacity. The track view {(b)} then displays these four tracks, whereas the data view {(c)} highlights the corresponding features (in magenta) along these tracks.
In particular, two features, $o_1$ and $o_2$ at time step $69$ are coupled with feature $o$ at time step $70$ with relatively high probabilities. By increasing the tracking probability threshold, shown in the box insert, feature $o_1$ will stop its track at time step $69$, whereas features $o_2$ and $o$ remain matched with each other. Our visual demo showcases such uncertainty in tracking. 

The visual demo is implemented using \textit{JavaScript} for the front-end, where the tracking graphs are visualized with \textit{D3.js} and the scalar fields are visualized using  \textit{WebGL}. The computational back-end is built with \textit{Python} and \textit{Flask}.

\section{Conclusion}
\label{sec:conclusion}

In this paper, we provide a flexible framework for tracking topological features in time-varying scalar fields. 
Our framework builds upon tools from topological data analysis (\ie, merge trees) and partial optimal transport. 
In particular, we model a merge tree as a measure network, and define a new partial fused Gromov-Wasserstein distance between a pair of merge trees. 
Such a distance gives rise to a partial matching between topological features in time-varying data, thus enabling flexible topology tracking for scientific simulations, as demonstrated by our extensive experiments.

On the other hand, our framework is not without limitations. 
First, we focus on feature tracking using merge trees, that is, we aim to preserve \emph{sublevel set} relations between features (\ie, critical points) that are captured by merge trees. Other topological descriptors such as Reeb graphs and Morse complexes may capture different topological relations such as \emph{level set} or \emph{gradient} relations. 
We would like to explore topology tracking with partial optimal transport using other types of topological descriptors, which are left for future work. 
Second, we provide experimental justifications for parameter tuning; understanding parameter tuning from a theoretical standpoint seems elusive. For future work, given the efficiency of our implementation, we would like to perform experiments involving datasets from large-scale simulations.

\ifCLASSOPTIONcompsoc
  \section*{Acknowledgments}
\else
  \section*{Acknowledgment}
\fi
This project was partially supported by DOE DE-SC0021015, NSF IIS 2145499, IIS 1910733, and DMS 2107808. 
\update{We would like to thank Andi Walther for sharing the {\CL} dataset and Dwaipayan Chatterjee for processing the dataset.} 


\ifCLASSOPTIONcaptionsoff
  \newpage
\fi

\bibliographystyle{IEEEtran}
\bibliography{GWMT-refs}
\vskip -2.2\baselineskip plus -1fil
\vspace{-2mm}
\begin{IEEEbiography}[\vspace{-6mm}{\includegraphics[width=1in,height=1.25in,clip,keepaspectratio]{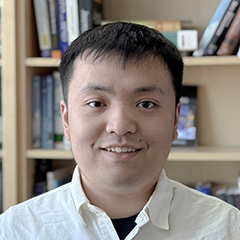}}]{Mingzhe Li} is a Ph.D. student at the School of Computing and the Scientific Computing and Imaging (SCI) Institute, University of Utah. His recent work combines optimal transport with topological data analysis in visualizing time-varying scientific simulations.  
\end{IEEEbiography}
\vskip -2.2\baselineskip plus -1fil
\vspace{-6mm}
\begin{IEEEbiography}[\vspace{-6mm}{\includegraphics[width=1in,height=1.25in,clip,keepaspectratio]{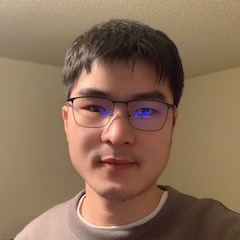}}]{Xinyuan Yan} is a Ph.D. student at the School of Computing and the SCI Institute, University of Utah. His current research focuses on visual analytics, set visualization, and explainable AI.
\end{IEEEbiography}
\vskip -2.2\baselineskip plus -1fil
\vspace{-6mm}
\begin{IEEEbiography}[\vspace{-6mm}{\includegraphics[width=1in,height=1.25in,clip,keepaspectratio]{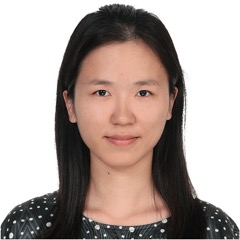}}]{Lin Yan} is an Assistant Professor of Computer Science at the Iowa State University. 
She received her Ph.D. in computing from the University of Utah in 2022. 
She studies complex data from scientific simulations by combining techniques from topological data analysis, statistics, machine learning, and visualization. 
\end{IEEEbiography}
\vskip -2.2\baselineskip plus -1fil
\vspace{-6mm}
\begin{IEEEbiography}[\vspace{-6mm}{\includegraphics[width=1in,height=1.25in,clip,keepaspectratio]{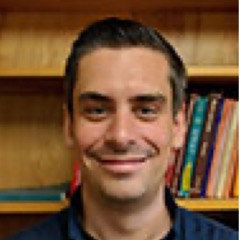}}]{Tom Needham} is  an  Assistant  Professor of Mathematics at Florida State University. He received his Ph. D. in Mathematics from University of Georgia. He works on applications of geometry and topology to problems in data science, computer vision and signal processing.
\end{IEEEbiography}
\vskip -2.2\baselineskip plus -1fil
\vspace{-6mm}
\begin{IEEEbiography}[\vspace{-6mm}{\includegraphics[width=1in,height=1.25in,clip,keepaspectratio]{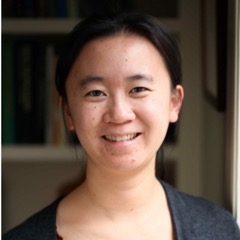}}]{Bei Wang} is  an  Associate  Professor  at  the  School  of  Computing and a faculty member at the SCI Institute,  University  of  Utah.  She  received  her  Ph.D. in  Computer  Science  from  Duke  University.  Her research  interests  include  topological  data  analysis,  data  visualization,  computational  and applied topology,  computational  geometry,  machine learning, and data mining. 
\end{IEEEbiography}
\clearpage
\appendices
\section{A New Stability Theorem for GW Distance}
\label{sec:stability-details}

In this section, we consider merge trees arising from functions on simplicial complexes and their associated measure network representations. Our goal is to prove the stability theorem, \autoref{thm:stability-main}, which is a result in the tradition of topological data analysis (\eg, \cite{cohen2007stability,chazal2009gromov,skraba2020wasserstein}). That is, we would like to show that a small change in the function data produces a small change in merge tree representations, as measured by the GW distance. We first remind the reader of technical details and notations before restating and proving the theorem.

Let $X$ be a finite, connected, geometric simplicial complex with a vertex set $V$. Let $f:X \to \R$ be a function obtained by starting with a function $f:V \to \R$ on the vertex set and extending linearly over higher dimensional simplices.
These are the types of functions we deal with in our visualization pipelines. Let $T_f$ denote the merge tree for $f$, defined as a quotient space $T_f = X/\sim$. Of course, when dealing with merge trees computationally, they are  represented by finite sets of points. A \emph{labeling} of a merge tree $T$ is a map from a finite set $\pi:S \to T$, which is surjective onto the set of leaf nodes of $T$ \cite{GasparovicMunchOudot2019}. Given a labeling, one can construct a \emph{least common ancestor matrix} $W$, indexed over $S \times S$, where $W(s,t)$ is the height of the highest point along the unique geodesic path from $\pi(s)$ to $\pi(t)$ in $T$. It is shown in \cite[Lemma 2.9]{GasparovicMunchOudot2019} that a merge tree with labeling is uniquely encoded by its least common ancestor matrix. We will specifically choose a labeling of $T_f$ given by restricting the quotient map $X \to T_f$ to the finite vertex set $V$. In the following part of this section, let $\pi_f:V \to T_f$ denote this labeling map and let  $W_f$ denote the least common ancestor matrix for this labeling. 

Let $p$ be a probability distribution over the vertex set $V$. We will assume that $p$ is \emph{balanced}, in the sense that for any $u,v,w \in V$, we have $p(u) \cdot p(v) \leq p(w)$; this property holds for the uniform distribution, for example. We then define the measure network representation of $T_f$ to be $G_f=(V,p,W_f)$, with $W_f$ denoting the least common ancestor matrix, as defined above. We can also define a family of weighted norms on the space of functions $f:V \to \R$ by 
\[
\|f\|_{L^q(p)} := \left(\sum_{v \in V} |f(v)|^q p(v)\right)^{1/q}.
\]

\begin{remark}
    The reader might observe that the vertex set of $G_f$ is the same as that of the mesh $X$. Generically, this will be the case for the merge tree associated to $f$, as a generic function will take distinct values on all vertices, so that no vertices are identified when constructing the merge tree. Even if vertices do get identified, the measure network representation $G_f$ will be \emph{weakly isomorphic} to the merge tree, in the sense that the GW distance between $G_f$ and the merge tree is zero---see Definition 2.3 and Theorem 2.4 of \cite{ChowdhuryMemoli2019}. This means that the measure network construction $G_f$ used in the stability result is a valid representation of the merge tree structure used throughout the paper.
\end{remark}

We now restate \autoref{thm:stability-main}.

\noindent\textbf{Theorem~\ref{thm:stability-main}.}
\emph{
Let $f,g:X \to \R$ be functions defined as above and let $p$ be a balanced probability distribution. Then}
\begin{equation}
\label{eqn:stability}
d^{GW}_q(G_f,G_g) \leq \frac{1}{2} |V|^{2/q} \|f - g\|_{L^q(p)}.
\end{equation}

The proof of the theorem will use Lemma \ref{lemma:lca}.

\begin{lemma}
\label{lemma:lca}
The least common ancestor matrix $W_f$ of $f:X \to \R$ can be expressed as
\[
W_f(v,w) = \min_{u_1,u_2,\ldots,u_n} \max_{u_i} f(u_i), 
\]
where the minimum is over paths $v = u_1,u_2,\ldots,u_n = w \in V$ where each pair of consecutive vertices is connected by an edge of $X$. 
\end{lemma}

\begin{proof}
We first show that 
\[
W_f(v,w) = \inf_{\gamma} \max_{t} f(\gamma(t)),
\]
where the infimum is over continuous piecewise linear paths $\gamma:[0,1] \to X$ with $\gamma(0) = v$ and $\gamma(1) = w$. Indeed, any path $\gamma$ from $v$ to $w$ projects under the quotient map to a continuous path $\bar{\gamma}:[0,1] \to T_f$ from $\pi_f(v)$ to $\pi_f(w)$ ($\pi_f$ denoting the restriction of the quotient map $X \to T_f$ to the vertex set, as above) with property that $f(\gamma(t)) = f(\bar{\gamma}(t))$ ($f$ denoting both the map $f:X \to \R$ and the induced map $f:T_f \to \R$, by abuse of notation). Then, we have
\[
\inf_{\gamma} \max_{t} f(\gamma(t)) = \inf_{\bar{\gamma}} \max_{t} f(\bar{\gamma}(t)) \geq \inf_{\eta} \max_t f(\eta(t)),
\]
where the latter infimum is over all continuous paths $\eta$ joining $\pi_f(v)$ to $\pi_f(w)$ in $T_f$. Since any such path will pass through the highest point along the unique geodesic path from $\pi_f(v)$ to $\pi_f(w)$, this infimum is equal to $W_f(v,w)$. Conversely, the geodesic path $\eta$ from $\pi(v)$ to $\pi(w)$ can be lifted to a piecewise linear path $\hat{\eta}$ joining $v$ to $w$ in $X$ with $f(\eta(t)) = f(\hat{\eta}(t))$ for all $t$. It follows that
\[
W_f(v,w) = \max_t f(\eta(t)) = \max_t f(\hat{\eta}(t)) \geq \inf_\gamma \max_t f(\gamma(t)).
\]

To complete the proof, we claim that any piecewise linear path $\gamma$ joining $v$ to $w$ in $X$ can be replaced by a path that passes only through edges of $X$, without increasing the maximum height along the path. Indeed, this can be done constructively. First, break the image of the path into a sequence of concatenated paths $\gamma_1, \gamma_2,\ldots,\gamma_n$, where each $\gamma_i$ is contained in the relative interior of a simplex $X_i$ of $X$. By construction, $X_i$ will be either a face of $X_{i+1}$, or vice versa. Moreover, $X_1 = v$ and $X_n = w$ (\ie, $\gamma_1$ and $\gamma_n$ are constant paths). To each $\gamma_i$, associate the vertex $u_i$ of $X_i$, which takes the lowest height. It follows that $u_i$ and $u_{i+1}$ are either equal or are connected by an edge for all $i$; we can ensure an edge path without repeats by simply removing repeated consecutive vertices. Since values of $f$ on $X$ are defined by linear interpolations of vertex values, $f(u_i)$ is at most the max value of $f$ over $\gamma_i$. This completes the proof of the claim. It follows that 
\[
W_f(v,w) = \inf_{\gamma} \max_{t} f(\gamma(t)) \geq \min_{u_1,u_2,\ldots,u_n} \max_{u_i} f(u_i).
\]
Since any edge path defines, in particular, a piecewise linear path, this inequality is actually an equality and it completes the proof of the lemma.
\end{proof}

\begin{proof}~(Proof of \autoref{thm:stability-main})
The proof follows from a sequence of inequalities
\begin{align}
    2d^{GW}_q(G_f,G_g) &\leq \|W_f - W_g\|_{L^q(p \times p)} \label{eqn:stability1} \\
    &\leq |V|^{2/q} \|(W_f - W_g) \ast (p \times p)^{1/q}\|_\infty \label{eqn:stability2} \\
    &\leq |V|^{2/q} \|(f-g) \ast p^{1/q}\|_\infty \label{eqn:stability3} \\
    &\leq |V|^{2/q} \|f - g \|_{L^q(p)}. \label{eqn:stability4}
\end{align}
We explain the notation and give a proof for each inequality below.

The right-hand side of \eqref{eqn:stability1} uses a weighted $L^q$ norm as defined above, for functions on $V \times V$. That is, for $F:V \times V \to \R$, we have
\[
\|F\|_{L^q(p \times p)} = \left(\sum_{v,w \in V} |F(v,w)|^q p(v)p(w)\right)^{1/q}.
\]
The inequality follows because the right-hand side is equal to the result of evaluating the $d^{GW}_q$ loss on the coupling induced by the identity map on $V$.

For $F:V \times V \to \R$, we define
\[
\left(F \ast (p \times p)^{1/q}\right)(v,w) := F(v,w)\left(p(v) \cdot p(w)\right)^{1/q}.
\]
Then, \eqref{eqn:stability2} follows from
\begin{align*}
\|W_f - W_g\|_{L^q(p\times p)} &= \|(W_f - W_g) \ast (p \times p)^{1/q}\|_{q} \\
&\leq |V|^{2/q} \|(W_f - W_g) \ast (p \times p)^{1/q}\|_\infty,
\end{align*}
where $\|\cdot\|_q$ denotes the standard $q$-norm (identifying the space of functions on $V\times V$ with $\R^{|V|^2}$) and the last inequality is a standard estimate on $q$-norms.

Similar to the above, for a function $h:V \to \R$, we define
\[
\left(h \ast p^{1/q}\right)(v) = h(v)p(v)^{1/q}.
\]
To establish \eqref{eqn:stability3}, we refer to Lemma \ref{lemma:lca}. Let $v,w \in V$ be arbitrary vertices and assume without loss of generality that $W_f(v,w) \geq W_g(v,w)$. We have 
\begin{align*}
    |W_f(v,w) - W_g(v,w)| &= W_f(v,w) - W_g(v,w) \\
    &= \min_{u_1,\ldots,u_n} \max_{u_i} f(u_i) - \min_{x_1,\ldots,x_n} \max_{x_i} g(x_i).
\end{align*}
Let $x_1^\ast,\ldots,x_n^\ast$ be a path realizing $W_g(v,w)$ and suppose that among these $x_i^\ast$'s, $f(x_i^\ast)$ is maximized at $x_k^\ast$. Then \[\min_{u_1,\ldots,u_n} \max_{u_i} f(u_i) \leq \max_{x_i^\ast} f(x_i^\ast) \leq f(x_k^\ast)\] and \[\min_{x_1,\ldots,x_n} \max_{x_i} g(x_i) = \max_{x_i^\ast} g(x_i^\ast) \geq g(x_k^\ast),\] hence
\begin{align*}
    \min_{u_1,\ldots,u_n} \max_{u_i} f(u_i) - \min_{x_1,\ldots,x_n} \max_{x_i} g(x_i) &\leq \max_{x_i^\ast} f(x_i^\ast) - \max_{x_i^\ast} g(x_i^\ast) \\
    &\leq f(x_k^\ast) - g(x_k^\ast).
\end{align*}
Using the assumption that $p$ is balanced, we have
\begin{align*}
    |(W_f - W_g)(p \times p)^{1/q}(v,w)| &\leq \left(f(x_k^\ast) - g(x_k^\ast)\right) \left(p(v) \cdot p(w)\right)^{1/q} \\
    &\leq \left(f(x_k^\ast) - g(x_k^\ast)\right) p(x_k^\ast)^{1/q} \\
    &\leq \max_{u \in V} |(f(u) - g(u))p(u)^{1/q}| \\
    &= \|(f-g)\ast p^{1/q}\|_\infty.
\end{align*}

Finally, \eqref{eqn:stability4} follows from a general bound on $q$-norms:
\begin{align*}
    \|(f-g) \ast p^{1/q}\|_\infty \leq \|(f - g) \ast p^{1/q}\|_q = \|f-g\|_{L^q(p)}.
\end{align*}
This equation completes the proof.
\end{proof}

Let us make some remarks about the result. First, if $p$ is chosen to be uniform, then the same proof gives an improved bound with a smaller power of $|V|$:
\[
d^{GW}_q(G_f,G_g) \leq \frac{1}{2} |V|^{1/q} \|f - g\|_{L^q(p)}.
\]

Next, the proof appears to be inefficient in that no optimization over couplings is used and that we pass through an $\infty$-norm rather than working directly with $q$-norms throughout. However, \eqref{eqn:stability} is asymptotically tight up to a universal constant, as is shown by the following example. Let $V = \{v_1,\ldots,v_{2n+1}\}$ and form $X$ by joining $v_i$ to $v_{i+1}$ by an edge for $i = 1,\ldots,2n$. Define functions $f,g:V \to \R$ by setting $g(v_i) = 0$ for all $i$ and 
\[
f(v_i) = \left\{\begin{array}{cl}
0 & i \neq n+1 \\
1 & i = n+1.
\end{array}\right.
\]
Let $p$ be a probability distribution on nodes with 
\[
p(v_i) = \left\{\begin{array}{cl}
p_0 & i \neq n+1 \\
p_1 & i = n+1,
\end{array}\right.
\]
where specific values of $p_0$ and $p_1$ are to be determined. Let $G_f(V,p,W_f)$ and $G_g(V,p,W_g)$ be the associated measure network representations of the merge trees.

\begin{proposition}
With $f$, $g$, $G_f$ and $G_g$ as above, there exist choices of $p_0$ and $p_1$ such that $p$ is a balanced distribution and
\[
\left(\frac{d_q^{GW}(G_f,G_q)}{\frac{1}{2}\|f - g \|_{L^q(p)}} \right)^q = C \cdot |V|^2 + c_{|V|}, 
\]
for a universal constant $C$, where $c_{|V|}$ represents a set of lower order terms in $|V|$. 
\end{proposition}

\begin{proof}
The coupling induced by the identity map is an optimal coupling of $G_f$ and $G_g$, hence
\begin{align*}
2d^{GW}_q(G_f,G_g) &= \|W_f - W_g\|_{L^q(p \times p)} = \|W_f\|_{L^q(p \times p)} \\
&= \left(\sum_{v,w} W_f(v,w)^q p(v)p(w) \right)^{1/q} \\
&= \left(2 p_0^2 n^2 + p_1^2 + 4 p_0 p_1 n \right)^{1/q},
\end{align*}
where the last quantity is obtained by counting entries equal to one in the least common ancestor matrix $W_f$. On the other hand,
\begin{align*}
    \|f - g\|_{L^q(p)} &= \|f\|_{L^q(q)} \\
    &= \left(\sum_{v} f(v)^q p(v) \right)^{1/q} = p_1^{1/q}.
\end{align*}
Since $|V| = 2n+1$, we have
\begin{align*}
 2 p_0^2 n^2 + p_1^2 + 4 p_0 p_1 n &= \frac{1}{2} |V|^2 p_0^2 + (p_0^2 + 2p_0 p_1)|V| \\
& \qquad \qquad \qquad + (\frac{1}{2}p_0^2 + p_1^2 + 2p_0 p_1)
\end{align*}
and it follows that 
\[
\left(\frac{d_q^{GW}(G_f,G_q)}{\frac{1}{2}\|f - g \|_{L^q(p)}} \right)^q = \frac{1}{2} |V|^2 \frac{p_0^2}{p_1} + c_{|V|}.
\]
The claim is proved (setting $C = \frac{1}{4}$) if we can choose $2p_0^2 = p_1$ such that $2n \cdot p_0 + p_1 = 1$ and $p_0,p_1 > 0$, so that $p$ defines a valid probability density. Solving the system of equations with the given constraints, we get
\[
p_0 = -n + \sqrt{n^2 + 2}.
\]
Moreover, we can check inequalities to show that the resulting probability density is balanced; the least obvious is $p_1^2 \leq p_0$, which holds as soon as $n \geq 2$.
\end{proof}

Finally, we prove Corollary \ref{cor:stability}, which we restate here for convenience.

\noindent\textbf{Corollary~\ref{cor:stability}.}
\emph{
Let $f,g:X \to \R$ be functions defined as above and let $p$ be a balanced probability distribution. Let $G_f$ (respectively, $G_g$) denote the representation of the merge tree $T_f$ ($T_g$) defined by the shortest path strategy. Then}
\begin{equation*}
d^{GW}_q(G_f,G_g) \leq \left(|V|^{2/q} + 2\right) \|f - g\|_{L^q(p)}.
\end{equation*}

\begin{proof}
Throughout this proof, we use $W_f$ and $W_g$ to represent least common ancestor matrices, and we use $D_f$ and $D_g$ to denote matrices of shortest path distances. Observe that for any $v,w \in V$ (the vertex set of $X$), 
\begin{equation}
\label{eqn:distance_formula}
D_f(v,w) = 2 W_f(v,w) - f(v) - f(w).
\end{equation}
Using notation from the proof of \autoref{thm:stability-main}, it follows that 
\begin{align}
    d^{GW}_q(G_f,G_g) &\leq \|D_f - D_g\|_{L^q(p \times p)} \label{eqn:cor_stab_1} \\
    &= \| (2 W_f - f - f) - (2 W_g - g - g)\|_{L^q(p \times p)} \label{eqn:cor_stab_2} \\
    &\leq 2 \|W_f - W_g\|_{L^q(p \times p)} + 2\|f - g\|_{L^q(p)} \label{eqn:cor_stab_3} \\
    &\leq |V|^{2/q} \|f-g\|_{L^q(p)} + 2\|f - g\|_{L^q(p)}, \label{eqn:cor_stab_4}
\end{align}
and the claim follows, once we justify the steps. \eqref{eqn:cor_stab_1} holds by the same reasoning as in the proof of \autoref{thm:stability-main}, \eqref{eqn:cor_stab_2} is an application of \eqref{eqn:distance_formula}, \eqref{eqn:cor_stab_3} is the triangle inequality for the $L^q$ norm, together with marginalizing out a copy of $p$ in the second term, and \eqref{eqn:cor_stab_4} follows from the proof of  \autoref{thm:stability-main}.
\end{proof}

\section{Experiments for the Stability Theorem}
\label{sec:stability-experiments}

In this section, we present experiments that demonstrate our new stability theorem for GW distance in Appendix~\ref{sec:stability-details}. 
First, we create a scalar field $f$ using a mixture of Gaussian functions. 
Second, we introduce different levels of noise $\iota$ to the scalar field, $f_{\iota}$.  
In particular, we add noise sampled uniformly from $\iota \in \{1\%, 2\%, \dots, 10\%\}$ of the range of $f$, see~\autoref{fig:stability-perturbations}.  
For each instance of $\iota$, we generate $20$ random scalar fields. 

We now compute the GW distance between the original scalar field $f$ and the noisy scalar fields $f_{\iota}$. 
We use the shortest path and lowest common ancestor   strategies for encoding edge information and uniform strategy for encoding the node information (so that the probability distribution over the vertex set is always \textit{balanced}). 

\begin{figure}[!ht]
    \centering
    \includegraphics[width=0.98\columnwidth]{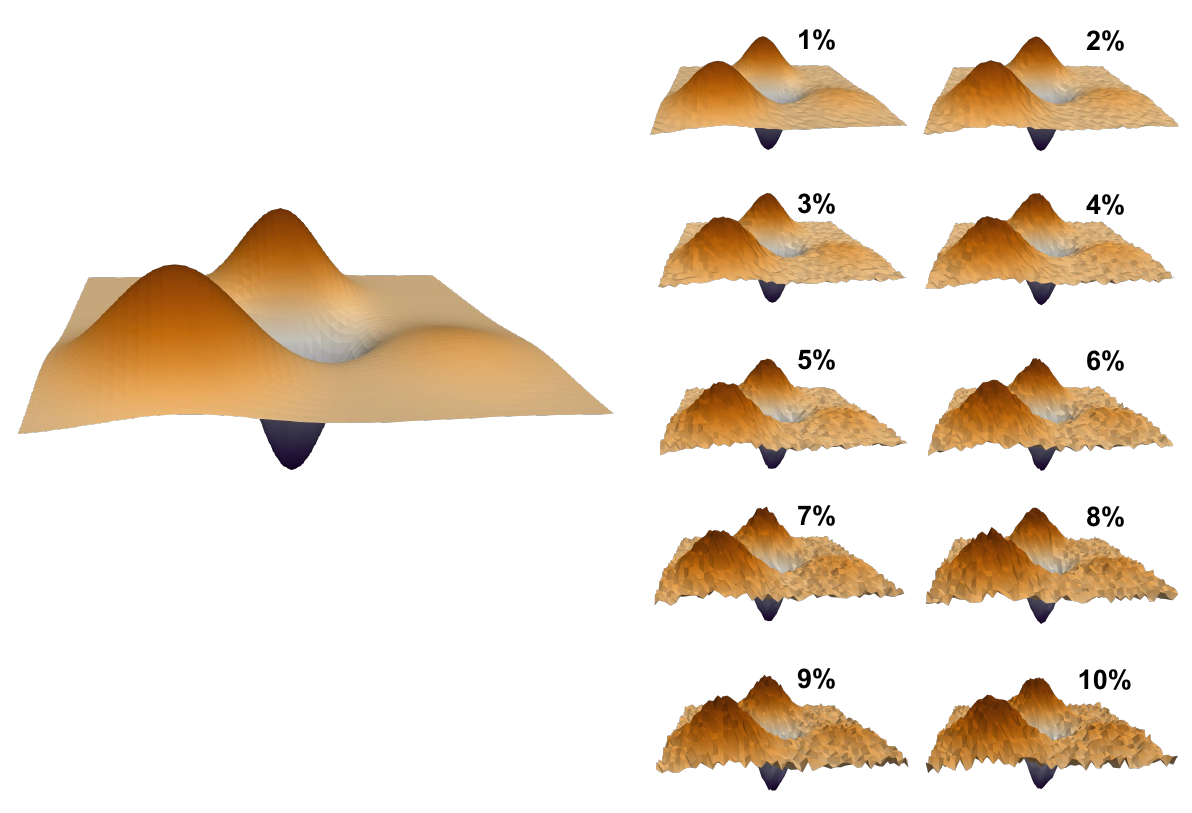}
    \vspace{-2mm}
    \caption{Left, a scalar field consists of multiple Gaussian functions. Right: ten datasets with different levels of noise.} 
    \label{fig:stability-perturbations}
\end{figure}

As stated in~\autoref{thm:stability-main}, when the probability distribution over the vertex set is \textit{balanced}, the upper bound for the GW distance using lowest common ancestor strategy is $\frac{1}{2} |V|^{2/q} \|f - g\|_{L^q(p)}$; referred to as the \emph{loose bound}. 
Furthermore, if the probability distribution on the node is uniform, there is a tighter upper bound of $\frac{1}{2} |V|^{1/q} \|f - g\|_{L^q(p)}$; referred to as the \emph{tight bound}.  

Our experimental results are shown in \autoref{fig:stability-statistics}. 
In (a), we compute the GW distance between $f$ and $f_{\iota}$ across varying $\iota$ using the lowest common ancestor strategy. 
For each fixed $\iota$, we compute the GW distances across $20$ instances and plot the corresponding box plot. 
We also plot the box plots associated with the tight and the loose bound, respectively. 
The line plots connect the mean values of these box plots. 
As shown in (a) and the zoomed-in view (b), the average GW distance (in blue) is upper bounded by the tight bound (in red), which is then upper bounded by the loose round (in yellow); thus validating~\autoref{thm:stability-main}.   

On the other hand, as shown in~\autoref{fig:stability-statistics}(c), using the shortest path strategy, the bound for the GW distance $\left(|V|^{2/q} + 2\right) \|f - g\|_{L^q(p)}$ is consistently higher than the GW distance itself in all our experiments; thus validating Corollary~\ref{cor:stability}. 

\begin{figure}[!ht]
    \centering
    \includegraphics[width=0.98\columnwidth]{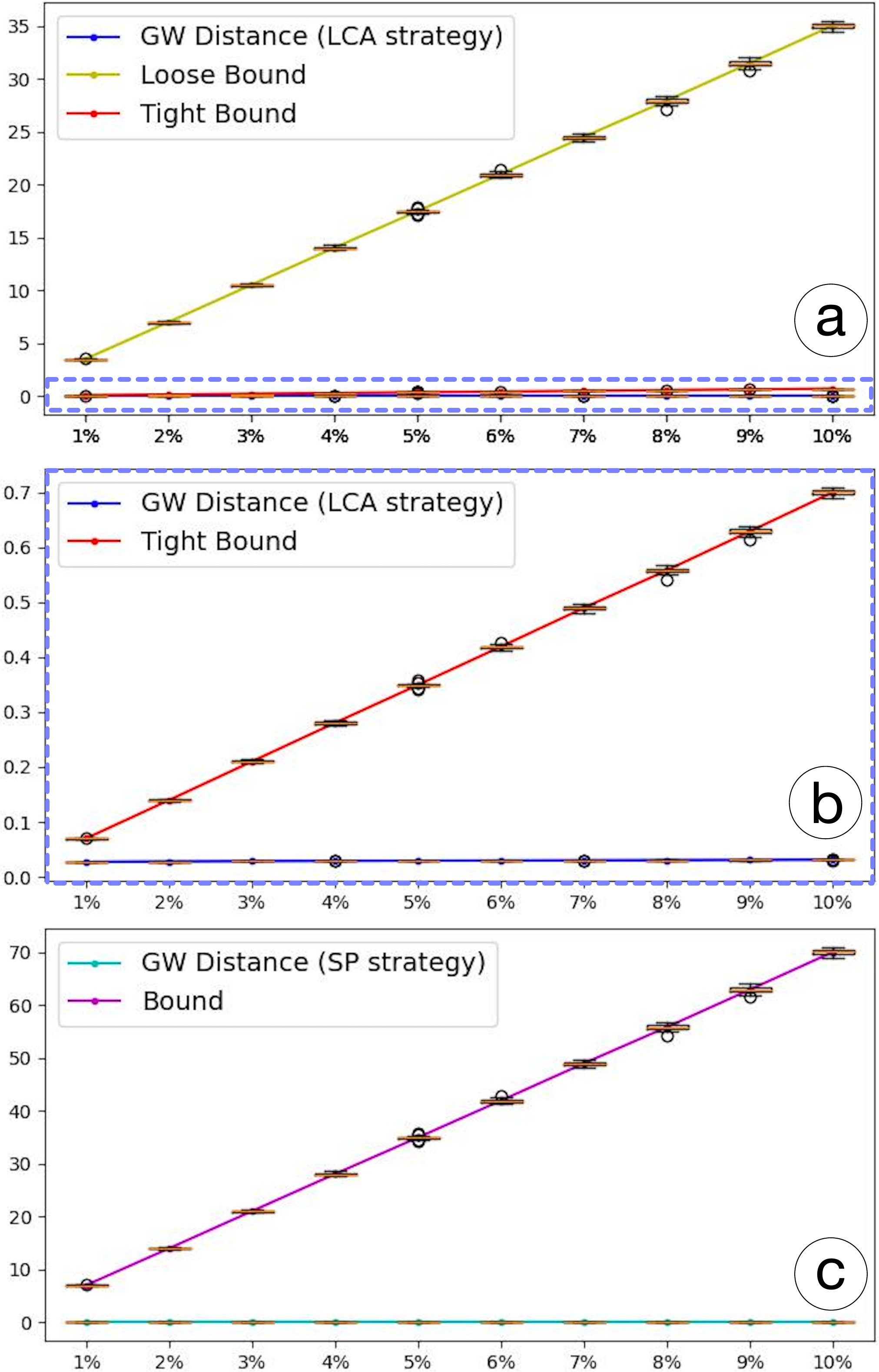}
    \vspace{-2mm}
    \caption{Box plots for the GW distance, the loose bound, and the tight bound under the lowest common ancestor (LCA) strategy (a-b), and the shortest path (SP) strategy (c).} 
    \label{fig:stability-statistics}
\end{figure}
\section{Parameter Tuning Details}
\label{sec:parameter-tuning-details}

We include details about parameter tuning for GFT and LWM described in~\autoref{sec:results} for reproducibility. 
We also provide all parameters across the three methods. 

\para{Parameters for {GFT}.}
$\lambda$ is a parameter used in {GFT} to balance the weight between region overlap size and histogram similarity. We follow the GitHub implementation (\url{https://github.com/hsaikia/mtlib}) to determine $\lambda$ using a binary search. 
Furthermore, GFT ignores matched feature pairs that have small region overlap and low histogram similarity (referred to as their \emph{combined similarity}). 
We sort all matched pairs across all time steps based on their combined similarity.  
We keep a certain percentage of the matched pairs based on the similarity.

\begin{figure*}[ht]
    \centering
    \includegraphics[width=0.98\textwidth]{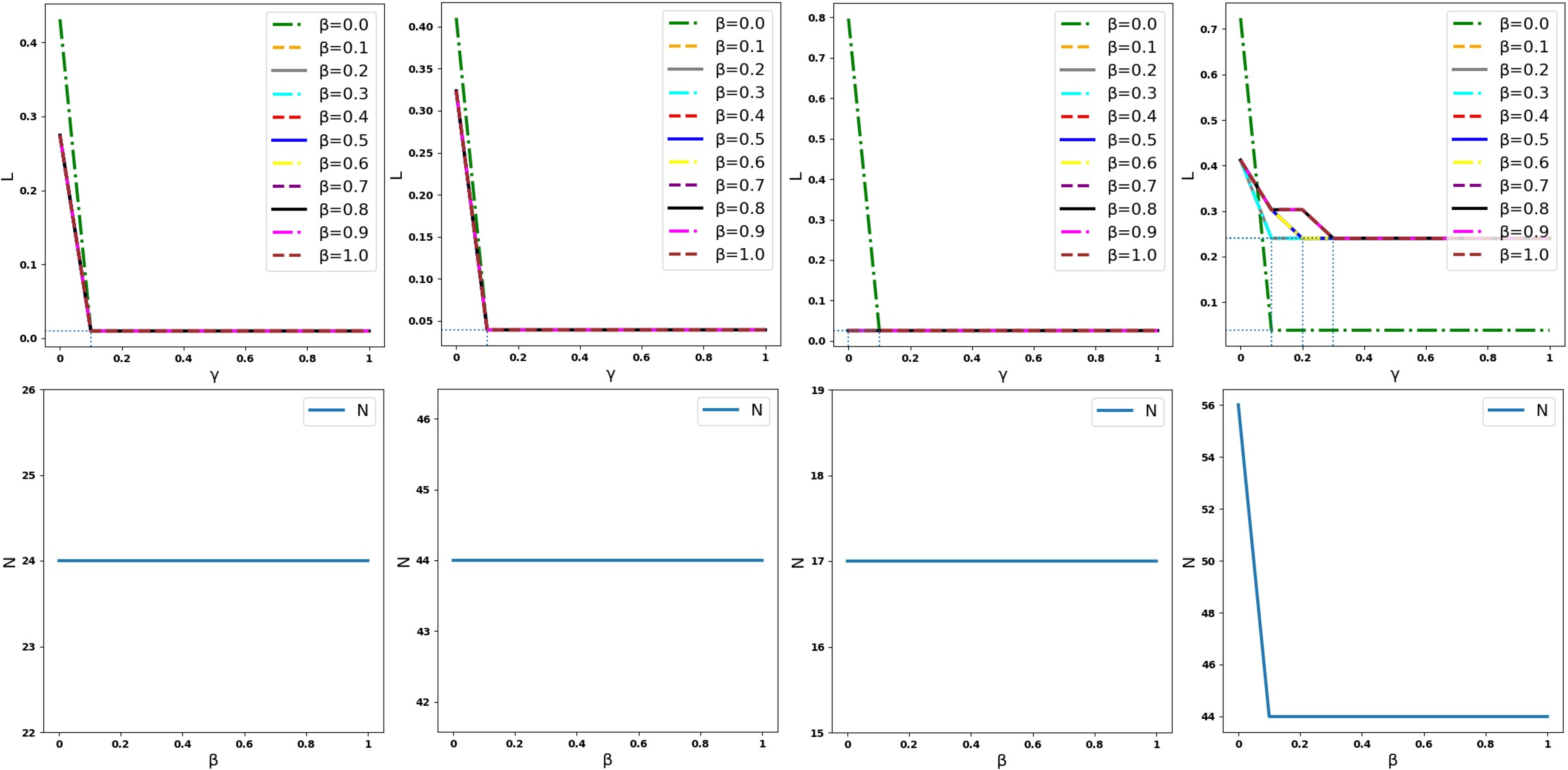}
    \vspace{-2mm}
    \caption{Parameter tuning of the {LWM} approach. From left to right: for the {\HC}, {\UCF}, {\VS}, and {\IF} datasets, respectively. For each dataset, the top shows the maximum matched distances $L$ at each choice of $\beta$ as $\gamma$ changes. We use the elbow points to find the optimal value of $\gamma$ that produces the lowest $L$. The bottom  shows the number of trajectories $N$ as $\beta$ changes with the optimal $\gamma$.} 
    \label{fig:lwm-tuning}
\end{figure*}

\para{Parameters for {LWM}.}
There are five parameters introduced in {LWM}: $\alpha$, $\beta$, $\gamma_x$, $\gamma_y$, $\gamma_z$. $\alpha$ and $\beta$ are weights associated with the Wasserstein distance, whereas $\gamma_x$, $\gamma_y$, $\gamma_z$ are the weights for geometric distances between critical points on the corresponding axes.
According to~\cite{SolerPlainchaultConche2018}, when tracking local maxima, there are guidelines to follow: $\gamma_x=\gamma_y=\gamma_z=\gamma$, $\alpha=0.1\beta$. Hence, we use grid search to determine $\gamma$ and $\beta$. 

\autoref{fig:lwm-tuning} evaluates  \emph{oversegmentations} and \emph{mismatches}   during the parameter tuning process. 
We first apply a grid search for $\gamma$ and $\beta$: fixing $\beta$ for each curve, we compute the maximum matched distance with different values of $\gamma$, and find the optimal value of $\gamma$ at the elbow points. 
Using these optimal $\gamma$, we then compute the number of trajectories {\wrt} $\beta$. 

Following this strategy for the {\HC}, {\UCF}, and {\VS} datasets, we set $\gamma_x=\gamma_y=\gamma_z=\gamma=0.1$, $\alpha=0.1$, $\beta=1$. 
However, for {\IF}, we notice that when $\beta \geq 0.1$, the maximum matched distance is much higher than the one with $\beta=0$, which implies significant mismatches when $\beta \geq 0.1$. Therefore, regardless of the number of trajectories, we set $\gamma_x=\gamma_y=\gamma_z=0.1$, $\alpha=\beta=0$. It is noticeable that in {LWM}, $\beta$ does not affect much on feature tracking for these four test datasets (see~\autoref{fig:lwm-tuning} bottom). The geometric information appears to be the dominant factor in matching.

\para{Parameter settings for all datasets.} 
We summarize parameter settings for all datasets. 

For {\HC} dataset, pFGW sets $\epsilon=6\%$, $\alpha=0.1$, $L^*=0.00997$; GFT retains $5\%$ of the matched pairs; LWM uses $\gamma=0.1$, $\beta=1$. 

For {\UCF} dataset, pFGW sets $\epsilon=1\%$, $\alpha=0.1$ and $L^*=0.03768$. 
For {\VS} dataset, pFGW sets $\epsilon=1\%$, $\alpha=0.1$, and $L^*=0.02537$.  
Both {\UCF} and {\VS} datasets use the same GFT and LWM configurations as {\HC}. 

For the {\IF} dataset: pFGW sets $\epsilon = 10\%$, $\alpha=0.4$, and $L^*=0.02693$; GFT retains $5\%$ of the matched pairs; LWM uses $\gamma=0.1$, $\beta=0.0$. 

For the {\IS} dataset: we did not include GFT results; LWM sets $\gamma=1$ and $\beta=1$. For pFGW, we use $\epsilon=10\%$, $\alpha=0.6$, $L^* = 0.4010$.  $L^*$ value is large because there are large gaps between time instances in this dataset and the hurricane also makes large movement across these time steps. 

For the {\CL} dataset: pFGW uses $\epsilon=50\%$, $\alpha=0.1$ and $L^*=0.0653$; 
GFT retains $2\%$ of all possible matches. 
After extensive parameter tuning, LWM sets $\beta=0.00$ and $\gamma=1.0$, which emphasize the importance of critical point location. 

For the {\VF} dataset: pFGW uses $\epsilon = 1\%$, $\alpha = 0.1$, and $L^* = 0.1369$. 
LWM sets $\beta = 0.00$ (producing the fewest isolated local maxima), and $\gamma=1.0$. 
GFT retains $10\%$ of all possible matches.

\begin{figure*}[!ht]
    \centering
    \includegraphics[width=0.98\textwidth]{subsampling.pdf}
    \vspace{-2mm}
    \caption{Tracking results of {\HC} (top), {\UCF} (middle), and {\IF} (bottom) dataset under subsampling. From left to right: the original pFGW, pFGW, LWM, and GFT with subsampling, respectively.} 
    \label{fig:subsampling}
\end{figure*}

\section{Subsampling and Robust Tracking}
For both the {\VS} and {\IS} datasets, we observe that topology-based feature tracking (such as pFGW and LWM) behaves better than geometry-based methods (such as GFT) when there are not sufficient region overlaps between adjacent time steps.  
In this section, we further examine the robustness of the three methods by subsampling time steps from previous datasets. 
We generate \emph{subsampled datasets} by sampling a single instance for every $6$, $15$, and $10$ time steps for {\HC}, {\UCF}, and {\IF} datasets, respectively. 

\subsection{Qualitative Comparisons}
The tracking results for these subsampled datasets are shown in~\autoref{fig:subsampling}, using the original tracking pFGW results (first column) as a reference.  
We expect the tracked trajectories to be similar for a robust tracking method, with and without subsampling. 

For the subsampled {\HC} dataset in~\autoref{fig:subsampling} (top), all three methods preserve the overall shape of trajectories, whereas pFGW demonstrates a slight advantage. 
In particular, some trajectories obtained by pFGW are missing by LWM (c.f., red boxes), whereas GFT produces oversegmentations (c.f., blue boxes). 

For the subsampled {\UCF} dataset in \autoref{fig:subsampling} (middle), LWM introduces obvious mismatches by matching geometrically distant critical points, whereas GFT creates a great number of broken trajectories on the left.
In comparison, pFGW produces better tracking results without oversegmentations or mismatches.

For the subsampled {\IF} dataset in \autoref{fig:subsampling} (bottom), all three methods show their limitations on tracking. LWM is able to preserve only a subset of long-term features and misses other features. 
GFT fails to preserve any major trajectories under subsampling. 
In comparison,  pFGW is able to replicate major patterns of the trajectories, especially the long-term ones. However, we can also see some mismatches in its tracking results. 

Based on these visualizations, pFGW is better than the other methods shown for preserving trajectories under subsampling while minimizing oversegmentations and mismatches. 
To evaluate these results quantitatively, we now discuss quantitative comparisons.   

\subsection{Quantitative Comparisons}
We utilize the notion of the Jaccard index to study the similarity between two sets of trajectories. 
Let $a$ and $b$ denote a pair of trajectories, each of which contains a finite number of critical points sampled at discrete time steps. 
We define the overlap between $a$ and $b$ as their Jaccard index, 
$$J(a, b)=\frac{|a \cap b|}{|a \cup b|}.$$

Let $A$ and $B$ be two sets of trajectories produced by two tracking methods, respectively. 
For each trajectory $a \in A$, define its matched trajectory 
$\pi(a) \in B$ such that $\pi(a) = \argmax_{b \in B} J(a, b).$ 
For any $a$, $\pi(a)$ may not be unique. 

We then introduce two measures that quantify the similarity between $A$ and $B$: 
$$S(A, B)=\frac{\sum_{a \in A}J(a, \pi(a))}{|A|}$$
$$S_W(A, B)=\frac{\sum_{a \in A}J(a, \pi(a)) |a|}{\sum_{a \in A} |a|}$$
$S$ captures the average overlap of all trajectories in $A$ against their matched ones in $B$, whereas $S_W$ is a weighted version of $S$ considering the lengths of trajectories in the summations. 
$S$ and $S_W$ are not symmetric and have optimal values of $1$ when $A = B$.     

In a subsampled dataset, a number of critical points may be missing from the original dataset. 
Let $A$ be the set of sub-trajectories from the original tracking results restricted to the subsampled time steps. 
Let $B$ be the set of trajectories obtained from the subsampled dataset. 
$S(A, B)$ and $S_W(A, B)$ describe how well a tracking algorithm preserves the trajectories against subsampling, whereas $S(B, A)$ and $S_W(B, A)$ indicate how well a tracking algorithm avoids mismatches in the subsampled dataset. 
In our experiment, we ignore (sub)trajectories of length $1$ as they are isolated critical points.  

\begin{table}[!ht]
\begin{center}
\resizebox{1.0\columnwidth}{!}{
\begin{tabular}{c|c|cccc}
\hline
{Dataset} & Method 
& $S(A, B)$ & $S(B, A)$ & $S_W(A, B)$ & $S_W(B, A)$\\ \hline
\multirow{3}{*}{\HC} 
& GFT & 0.763 & 0.806 & 0.804 & 0.828  \\
& LWM & 0.877 & \textbf{0.926} & 0.930 & \textbf{0.967}           \\
& pFGW & \textbf{0.902} & 0.902 & \textbf{0.944} & 0.940  \\ \hline
\multirow{3}{*}{\UCF}          
& GFT & 0.670 & 0.858 & 0.716 & 0.839           \\
& LWM  & 0.344 & 0.587 & 0.488 & 0.617           \\
& pFGW & \textbf{0.907} & \textbf{0.990} & \textbf{0.957} & \textbf{0.991} \\ \hline
\multirow{3}{*}{\IF}  
& GFT & 0.314 & 0.452 & 0.228 & 0.419           \\
& LWM & 0.341 & 0.556 & 0.413 & 0.535           \\
& pFGW & \textbf{0.552} & \textbf{0.624} & \textbf{0.588} & \textbf{0.598} \\ \hline
\end{tabular}
}
\end{center}
\caption{Similarity measures between a pair of tracking results with and without subsampling. For a fixed tracking method, $A$ denotes the trajectories restricted to subsampled time steps. $B$ denotes the trajectories from the subsampled dataset. The highest scores are in bold.}
\label{table:subsampling}
\vspace{-6mm}
\end{table}

The quantitive evaluation results are provided in~\autoref{table:subsampling}. 
pFGW is shown to have better performance than GFT and LWM in terms of capturing original trajectories under subsampling in almost all cases. 
In particular, for the {\UCF} and {\IF} datasets, pFGW obtains significantly higher similarity measures than GFT and LWM. 
These results align well with our observations in~\autoref{fig:subsampling} that pFGW is better at preserving trajectories and avoiding mismatches for subsampled datasets.

Drawbacks of GFT and LWM are also evident in~\autoref{table:subsampling}. 
For GFT, $S(A,B)$ and $S_W(A, B)$ over the {\UCF} dataset are low because GFT fails to maintain continuity of trajectories on the left. 
For the {\IF} dataset, GFT does not maintain long-term trajectories, leading to low similarity measures. 
For LWM, similarity measures are lowest for the {\UCF} dataset due to significant mismatches in the tracking results. 
LWM maintains only a few long-term features for the {\IF} dataset, leading to measures lower than those from pFGW. 

To summarize, based on both qualitative and quantitative evaluations, GFT appears to lose its ability to track features when there are not sufficient time resolutions for geometry-based tracking, for instance, the subsampled {\IF} dataset. 
Whereas LWM captures major features during tracking, it is not as robust as pFGW in tracking features for datasets with low time resolutions. 
For example, for the subsampled {\IF} datasets, LWM misses a large portion of the original trajectories. 
For the {\UCF} dataset, LWM generates many obvious mismatches. Such drawbacks are also clearly reflected in the similarity measures. 
In comparison, our pFGW method performs quite well in robustly tracking features on datasets with low time resolutions.


\end{document}